\documentclass[11pt, a4paper]{article}

\usepackage[margin=1in]{geometry}
\usepackage{amsmath, amssymb, amsthm}
\usepackage{enumitem}
\usepackage{booktabs}
\usepackage{array}
\usepackage{xcolor}
\usepackage{graphicx}
\usepackage{hyperref}
\usepackage[expansion=false]{microtype}
\usepackage{algorithm}
\usepackage{algpseudocode}
\usepackage{mathtools}

\hypersetup{colorlinks=true, linkcolor=blue!60!black, citecolor=blue!60!black, urlcolor=blue!60!black}
\usepackage[nameinlink,capitalise]{cleveref}

\newtheorem{proposition}{Proposition}

\newtheorem{corollary}{Corollary}

\newtheorem{assumption}{Assumption}

\newtheorem{desideratum}{Desideratum}
\usepackage{libertine}
\newcommand{\E}{\mathbb{E}}
\newcommand{\Var}{\mathrm{Var}}
\newcommand{\Cov}{\mathrm{Cov}}

\newcommand{\QSMF}{\textsc{QS-MF}}

\title{Quality-Sensitive Matrix Factorization for Community Notes:\\ Towards Sample Efficiency and Manipulation Resistance}
\author{
\begin{tabular}{ccc}
Mohak Goyal$^{\ast}$ &
Nishka Arora$^{\ast}$ &
Ashish Goel \\
\nolinkurl{goyalmohak2@gmail.com} &
\nolinkurl{nishkaa@stanford.edu} &
\nolinkurl{ashishg@stanford.edu}
\end{tabular}\\[0.75em]
Department of Management Science \& Engineering, Stanford University
}
\date{}

\begin{document}
\maketitle
\begingroup
\renewcommand{\thefootnote}{\fnsymbol{footnote}}
\footnotetext[1]{Equal contribution.}
\endgroup
\setcounter{footnote}{0}

\begin{abstract}
Community Notes is X's crowdsourced fact-checking program: contributors write
short notes that add context to potentially misleading posts, and other
contributors rate whether those notes are helpful. Its algorithm
uses a matrix factorization model to separate ideology from note
quality, so notes are surfaced only when they receive support across
ideological lines. After ideology is accounted for, however, the model gives
all raters equal influence on quality estimates. This slows consensus
formation and leaves the quality estimate vulnerable to noisy or strategic
raters. We propose Quality-Sensitive Matrix Factorization (\QSMF{}), which uses a per-rater quality-sensitivity
parameter \(\hat\rho_i\) estimated jointly with all other parameters. This connects
\QSMF{} to peer prediction: without external ground truth, it gives more
influence to raters whose ideology-adjusted ratings are more consistent with
the note-quality estimates learned from all the ratings.

We evaluate \QSMF{} on 45M ratings over 365K notes from the six months before
the 2024 U.S.\ presidential election. Split-half tests confirm that quality
sensitivity is a stable, empirically recoverable rater trait. In evaluation on high-traffic notes, \QSMF{} requires 26--40\%
fewer ratings to match the baseline's accuracy. In semi-synthetic coordinated
attacks on notes of opposing ideology, \QSMF{} substantially reduces displacement on the estimated quality estimates of targeted notes relative
to the baseline. In synthetic data with known ground truth, \(\hat\rho_i\)
separates good from bad raters with an AUC above 0.94, and achieves much lower error in recovering the true note quality estimates in the presence of bad raters.
These gains come from a single additional scalar parameter per rater, with no
external ground truth and no manual moderation.

\end{abstract}

\section{Introduction}\label{sec:intro}

Misinformation spreads faster than it can be corrected. On social media platforms, a false claim can reach millions of users within hours, while professional fact-checking organisations, constrained by limited staff, can assess only a small fraction of the content that circulates online \cite{allen2021scaling}. By the time a professional verdict is published, the damage is often done. This fundamental asymmetry between the speed of misinformation and the capacity of expert review has motivated a search for scalable alternatives.
Crowdsourced fact-checking offers a fundamentally different approach. Rather than relying on a small cadre of professionals, it enlists ordinary users, the same crowd that consumes and shares content, to evaluate its accuracy. A growing body of evidence supports this strategy: politically balanced groups of laypeople can identify misinformation with accuracy comparable to professional fact-checkers \cite{allen2021scaling,martel2023crowds}, and crowdsourced corrections influence user behaviour, increasing the likelihood that authors retract misleading posts \cite{renault2024collaboratively}.

Community Notes (CN), originally launched as Birdwatch on Twitter in 2021 \cite{wojcik2022,prollochs2022birdwatch}, is the most prominent implementation of crowdsourced fact-checking at scale. Contributors write short notes that add context to potentially misleading posts; other contributors then rate these notes for helpfulness. The system's main innovation is its consensus mechanism: rather than simple majority vote, CN uses a ``bridging" algorithm that requires agreement across ideological lines \cite{ovadya2023}. A note is deemed helpful only when it receives support from raters who typically disagree. This design makes it difficult for any single partisan group to control which notes are displayed.

The approach has been quite successful. Displayed notes reduce engagement with misinformation \cite{slaughter2025engagement} and are trusted more than platform-authored labels across partisan lines \cite{drolsbach2024trust}. The model has been adopted beyond X: YouTube \cite{youtube_cn}, TikTok \cite{tiktok_cn}, and Meta \cite{meta_cn} have all launched Community Notes-style programmes. The open-source algorithm and publicly available rating data \cite{communitynotes_data} make CN a uniquely transparent content moderation system: a rare setting where the full pipeline from raw ratings to displayed labels can be independently scrutinised.

Underlying CN is a challenge that is fundamental to all such systems: how to aggregate judgments from many agents when there is no objective ground truth against which to verify them. The classical wisdom-of-crowds result, dating to Condorcet's jury theorem \cite{condorcet1785}, establishes that aggregating independent judgments from individually competent agents yields a collective decision that improves with group size \cite{surowiecki2004}. But the theorem's assumptions (equal competence, independence, and a single dimension of truth) rarely hold in practice, especially in politically charged settings where agents differ systematically in both ability and motivation.

When ground truth is unavailable or too costly to obtain, the \emph{peer prediction} literature provides the theoretical foundation for extracting signal from subjective reports. Miller, Resnick, and Zeckhauser \cite{miller2005} showed that truthful reporting can be incentivised by scoring each agent's report against those of peers, exploiting the statistical structure of inter-rater agreement. 
In the crowdsourcing literature, this principle has been operationalised through models that jointly estimate worker reliability and item difficulty from observed labels alone. Dawid and Skene \cite{dawid1979} introduced per-worker confusion matrices estimated via an EM-based algorithm. 
Raykar et al.\ \cite{raykar2010} developed a Bayesian framework for learning from crowds of annotators with varying reliability. The central insight across this literature is that not all agents are equally informative, and the data itself reveals whose rating is more reliable.

CN builds on this tradition but introduces a critical complication: the items being rated (notes often about politically charged content) evoke strong ideological reactions that are confounded with quality judgments. Standard crowdsourcing reliability models cannot disentangle these two channels. A partisan rater whose ratings are predictable and low-variance appears ``reliable'' in any model that equates consistency with quality. The CN  algorithm addresses this partially by modelling ideology as a separate dimension, encouraging that notes require genuine cross-partisan support. But while it correctly handles the \emph{direction} of ideological disagreement, it treats all raters as equally informative about note quality once ideology is accounted for.

This worsens well-documented vulnerabilities. Consensus emerges too slowly: publication delays often exceed the critical window for viral misinformation \cite{Razuvayevskaya2025_CommunityNotesConsensus,renault2024collaboratively}. Low-effort and strategically motivated raters distort quality estimates \cite{truong2025community}. These problems share a common root: \emph{influence on note quality is fixed rather than responsive to demonstrated quality tracking}. We show that learning a per-rater quality-sensitivity parameter which governs how much each rater's opinion influences note quality estimates, addresses both problems simultaneously, yielding faster convergence and greater robustness to adversarial manipulation.

\subsection{The Baseline Model and its Limitations}

We briefly introduce the CN model here, and a full treatment appears in Section~\ref{sec:baseline-model}. Each observed rating $r_{ij}$ (the score that rater $i$ assigns to note $j$) is modelled as the sum of four components:
\begin{equation}\label{eq:baseline}
    r_{ij} = \mu + \alpha_i + \beta_j + \gamma_i \delta_j + \epsilon_{ij}.
\end{equation}
Here $\mu$ is a global intercept capturing the average rating level; $\alpha_i$ is a rater-specific bias (some raters are systematically generous, others harsh); $\beta_j$ is the \emph{note quality} we ultimately want to estimate; and $\gamma_i \delta_j$ is a bilinear term that captures ideological alignment between rater $i$ and note $j$ along a latent one-dimensional ideology axis. The noise term $\epsilon_{ij}$ accounts for everything else. Because the product $\gamma_i \delta_j$ is unchanged if both coordinates are multiplied by $-1$, the sign of this latent axis is arbitrary: positive and negative values should not be interpreted as named political camps. The key design insight is the $\gamma_i \delta_j$ term: it implements bridging by ensuring that a note receives a high quality estimate $\hat\beta_j$ only when raters with very different $\gamma_i$ values agree it is helpful. 

Consider how the model estimates note quality. Once all other parameters are fitted, we can strip away the components that have nothing to do with quality by defining the \emph{de-ideologised rating}
\begin{equation}\label{eq:deideologized-intro}
    d_{ij} = r_{ij} - \hat\mu - \hat\alpha_i - \hat\gamma_i \hat\delta_j,
\end{equation}
which removes the rater's baseline tendency and ideological component, leaving behind the quality signal plus noise. The estimated note quality is then simply the average of these de-ideologised ratings: $\hat\beta_j = \sum_i d_{ij} / (n_j + \lambda_\beta)$, where every rater receives equal weight. Here $\lambda_\beta$ is the regularization weight. The problem is not ideology itself; ideological raters who also attend to quality contribute useful $d_{ij}$. The problem is raters whose signal is \emph{entirely} non-quality. 

\subsection{How \QSMF{} Works}\label{sec:how-qsmf}

Our method, \textbf{Quality-Sensitive Matrix Factorization (\QSMF{})}, replaces the additive note-quality term $\beta_j$ in the baseline model with the bilinear term $\rho_i \beta_j$, where $\rho_i \geq 0$ is a per-rater \emph{quality sensitivity} parameter estimated jointly with all other parameters:
\begin{equation}\label{eq:qsmf-intro}
    r_{ij} = \mu + \alpha_i + \rho_i \beta_j + \gamma_i \delta_j + \epsilon_{ij}.
\end{equation}
The standard CN model is the special case $\rho_i = 1$ for all $i$. A diagnostic based on the Rao score \cite{rao1948large} rejects this restriction on real data, confirming that quality sensitivity varies meaningfully across raters.

The bilinear structure $\rho_i \beta_j$ gates each rater's influence on $\hat\beta_j$ through their estimated $\hat\rho_i$. Through the normal equations, rater $i$'s contribution to $\hat\beta_j$ scales as $\hat\rho_i$, and their contribution to the ideology estimate $\hat\delta_j$ scales as $\hat\gamma_i$. The model provides \emph{channel-specific weighting}: each rater's information flows to the channel where it is informative.

Intuitively, $\rho_i$ answers the question: \emph{how much better do we understand this rater's ratings once we know note quality?} Every rater's ratings are partially explained by their baseline tendency ($\mu + \alpha_i$), their ideological alignment ($\gamma_i \delta_j$), and noise. What remains after removing baseline and ideology (the de-ideologized rating $d_{ij}$) is either informative about note quality or it is pure noise. The parameter $\rho_i$ captures exactly this distinction: the fraction of de-ideologized variance that the quality signal $\rho_i \beta_j$ explains beyond what baseline, ideology, and noise already account for.

\subsection{Contributions}

A growing body of work has documented the promise and fragility of CN: publication delays that exceed the window for viral content \cite{Razuvayevskaya2025_CommunityNotesConsensus,renault2024collaboratively}, estimate instability under rating surges \cite{chuai2026consensus}, outsized influence of hyperactive minorities \cite{nudo2026hyperactive}, vulnerability to bias and strategic manipulation \cite{truong2025community}, and systematic heterogeneity in rater quality \cite{chuai2025request, wack2026laziness}. However, to our knowledge, we provide the first algorithmic solution that \emph{simultaneously} improves sample efficiency and robustness to adversarial attacks, through a single, principled extension to the model. 

\begin{enumerate}[leftmargin=*]
    \item \textbf{Evaluation framework.} We formalise the desiderata of a crowdsourced fact-checking system operating without ground truth (sample efficiency, robustness, rater identification, and discriminative power; Section~\ref{sec:desiderata}), and use them as standing criteria throughout the paper.

    \item \textbf{\QSMF{}.} We introduce Quality-Sensitive Matrix Factorization, a single-parameter extension of the CN model that makes each rater's influence on note quality proportional to their demonstrated quality tracking. The per-rater parameter $\rho_i$ is estimated natively inside the matrix factorization.

    \item \textbf{Theoretical analysis.} We provide theoretical intuition for why \QSMF{} helps: the heterogeneous weighting is more sample-efficient for $\hat\beta_j$ estimation (Proposition~\ref{prop:efficiency}) under our assumed DGP, and the reputation mechanism structurally attenuates naive adversarial attacks (Section~\ref{sec:rho-update}).

    \item \textbf{Empirical evaluation.} 
 We evaluate on all the ratings on X Corp. from the six months prior to  the 2024 U.S.\ presidential election, complemented by synthetic experiments with known ground truth.  The results are replicated across repeated random splits, random seeds, rolling windows, and bootstrap resamples. \QSMF{} improves over the baseline on every desideratum (Sections~\ref{sec:oos}--\ref{sec:attacks}).
 
 We validate that \QSMF{} better describes the rating process than the baseline, using both out-of-sample prediction (Section~\ref{sec:oos}) and a score-based diagnostic that rejects the restriction $\rho_i = 1~\forall~i$ (Section~\ref{sec:dgp-validation}).
\end{enumerate}

\subsection{Evaluation Without Ground Truth}\label{sec:desiderata}

There is no ground truth for note quality $\beta_j$, and it is difficult to assign truth labels without receiving criticism from one or both political camps. Any evaluation must rely on internal consistency and operational performance. We define four properties that any rating aggregation algorithm should satisfy. 

\begin{desideratum}[Sample efficiency]\label{des:sample}
    Note quality estimates should converge quickly as ratings accumulate. Suppose each rating for note $j$ is drawn i.i.d.\ from a population of raters. Holding the ratings for all other notes constant, let $\hat\beta_j^{(k)}$ denote the estimate for the quality of note $j$ from $k$ ratings. Let $\hat\beta_j = \lim_{k \to \infty} \hat\beta_j^{(k)}$ denote the quality estimate in the large-sample limit. We require that $\E[(\hat\beta_j^{(k)} - \hat\beta_j)^2]$ decreases rapidly with $k$. 
\end{desideratum}

\begin{desideratum}[Robustness to manipulation]\label{des:robust}
    If a fraction $f$ of raters submit strategic ratings for a set of targetted notes, the quality estimate $\hat\beta_j$ for targeted notes should remain close to the estimate computed from honest raters alone. We measure robustness by the mean squared displacement $(\hat\beta_j^{\text{attacked}} - \hat\beta_j^{\text{clean}})^2$ on targeted notes. 
\end{desideratum}

\begin{desideratum}[Rater identification]\label{des:identification}
    If raters differ in a trait that affects the informativeness of their ratings for note quality, the algorithm should recover that trait from observed ratings alone. 
\end{desideratum}

\begin{desideratum}[Discriminative power]\label{des:discriminative}
    Suppose each note $j$ has a true, unobserved quality $\beta_j^*$. The algorithm's estimate should satisfy $\E[(\hat\beta_j - \beta_j^*)^2] \to 0$ as the number of ratings grows. 
\end{desideratum}

 Discriminative power is a necessary complement to the others: an algorithm that sets $\hat\beta_j = c$ for all $j$ trivially achieves perfect sample efficiency and perfect robustness, but is useless because it learns nothing about the true quality. These four criteria are jointly necessary and individually insufficient. They structure the entire experimental evaluation that follows.

\subsection{Relationship to Prior Work}
The problem of aggregating judgments from annotators of varying reliability has a long history. Dawid and Skene \cite{dawid1979} introduced the EM-based confusion-matrix model for discrete labels, and Raykar et al.\ \cite{raykar2010} extended this framework to continuous annotations with per-annotator variance, establishing conditions for unbiased estimation. GLAD \cite{whitehill2009} jointly models worker ability and task difficulty, and subsequent work has generalised along several axes: multidimensional annotator models \cite{welinder2010}, mixture models for annotator subpopulations \cite{guan2018}, and end-to-end integration with deep classifiers \cite{rodrigues2018}. In the matrix factorization literature, Salakhutdinov and Mnih \cite{salakhutdinov2008} introduced the probabilistic MF (PMF) framework on which the CN model is built, and Koren et al.\ \cite{koren2009matrix} established the standard techniques for recommender systems. 

Structurally, the quality channel of \QSMF{} resembles item response theory (IRT) \cite{baker2004,embretson2000}, where responses depend on the interaction of (possibly multi-dimensional) user and item features. We posit that quality-tracking is a meaningful and stable user feature, and should be included in the model. The ideology factor $\gamma_i\delta_j$ acts as a one-dimensional spatial voting model \cite{clinton2004}. The ideological factor is the key distinction from both standard IRT and the standard crowdsourcing setup. In models like Dawid--Skene, GLAD, or Raykar et al., a reliable annotator is one who agrees with the consensus. In CN, a rater can agree with the consensus for ideological reasons that carry no quality information. Only by jointly modelling the ideology channel can we separate ideology-driven consistency from genuine quality tracking (see Section~\ref{sec:how-qsmf}).

Peer prediction mechanisms \cite{miller2005,prelec2004,dasgupta2013,shnayder2016} reward reports that are informative about a latent truth without access to ground-truth labels. The $\hat{\rho}_i$ update in \QSMF{} is analogous to a peer prediction scoring rule: it rewards alignment with the emerging quality consensus after partialling out ideology, connecting MF optimisation to mechanism-design intuitions. 


\subsection{Roadmap}
Section~\ref{sec:baseline-model} reviews the standard CN model, discusses the role of regularisation, and explains why two natural enhancements fail (Section~\ref{sec:simpler-fixes}). Section~\ref{sec:stage2} introduces the \QSMF{} model: its data generating process, objective function, identifiability properties, and estimation procedure. Section~\ref{sec:theory} develops the theoretical properties of \QSMF{}, including efficiency gains from heterogeneous weighting (Section~\ref{sec:efficiency}), the interpretation of $\rho_i$ as quality-channel fit (Section~\ref{sec:rho-update}), and the structural attenuation of naive attacks (Section~\ref{sec:rho-update}). The experimental evaluation begins with data and training details (Sections~\ref{sec:data}--\ref{sec:training}), followed by DGP validation confirming that quality sensitivity is a real and stable rater trait (Section~\ref{sec:dgp-validation}). We then evaluate \QSMF{} against the baseline on out-of-sample prediction (Section~\ref{sec:oos}), sample efficiency (Section~\ref{sec:sample-efficiency}), attack robustness (Section~\ref{sec:attacks}), and synthetic ground-truth recovery (Section~\ref{sec:synthetic}). Section~\ref{sec:limitations} discusses limitations and directions for future work.

\section{The Standard Community Notes Model}\label{sec:baseline-model}

\paragraph{Notation.}
Let $\mathcal{I}$ denote the set of raters and $\mathcal{N}$ the set of notes. The set of observed rating pairs is $\Omega \subseteq \mathcal{I} \times \mathcal{N}$. For each note $j$, let $\mathcal{I}_j = \{i : (i,j) \in \Omega\}$ denote the set of raters who rated it, with $n_j = |\mathcal{I}_j|$. For each rater $i$, let $\mathcal{N}_i = \{j : (i,j) \in \Omega\}$ denote the set of notes they rated.

The standard CN model \cite{wojcik2022} assumes ratings are generated by
\begin{equation}\label{eq:baseline-dgp}
    r_{ij} = \mu + \alpha_i + \beta_j + \gamma_i \delta_j + \epsilon_{ij}, \quad \epsilon_{ij} \overset{\mathrm{iid}}{\sim} \textsc{Normal}(0, \sigma^2),
\end{equation}
where $\mu$ is a global intercept, $\alpha_i$ is rater bias, $\beta_j$ is note quality, and $\gamma_i \delta_j$ captures ideological alignment along a latent one-dimensional ideology axis. The model is fit by minimising the regularised loss
\begin{equation}\label{eq:baseline-objective}
\mathcal{L}(\theta) = \frac{1}{2}\sum_{(i,j) \in \Omega} \left(r_{ij} - \mu - \alpha_i - \beta_j - \gamma_i \delta_j\right)^2
+ \frac{\lambda_\alpha}{2}\sum_{i \in \mathcal{I}} \alpha_i^2
+ \frac{\lambda_\beta}{2}\sum_{j \in \mathcal{N}} \beta_j^2
+ \frac{\lambda_\gamma}{2}\sum_{i \in \mathcal{I}} \gamma_i^2
+ \frac{\lambda_\delta}{2}\sum_{j \in \mathcal{N}} \delta_j^2.
\end{equation}
A standard connection between penalised likelihood and Bayesian inference \cite{hoerl1970,bishop2006} implies that each $L_2$ penalty is equivalent to a zero-mean Gaussian prior on the corresponding parameter in a MAP estimation framework. Regularisation serves three purposes: it resolves the scale and rotation ambiguities inherent in the bilinear term $\gamma_i \delta_j$ by selecting a unique factorisation; it controls variance of the parameter estimates, which is critical for matrix factorization over sparse matrices; and it provides sensible defaults for parameters with limited data. In particular, a new rater's ideology $\hat\gamma_i$ is shrunk toward zero, meaning the model treats them as ideologically neutral until enough ratings accumulate to estimate their position.

This prior is both a strength and a vulnerability. Without it, ideology estimates for low-activity raters would be noise-dominated, making the algorithm unstable. But the algorithm rewards agreement from ideologically \emph{diverse} raters, and a rater with $\hat\gamma_i \approx 0$ appears to be exactly the kind of neutral evaluator whose rating affects the note quality most. 
Recall the de-ideologized rating \(d_{ij}\) from
\eqref{eq:deideologized-intro}. Holding the ideology coordinate fixed, the note quality is estimated as:
\begin{equation}\label{eq:beta-baseline}
    \hat{\beta}_j^{\text{base}} = \frac{\sum_{i \in \mathcal{I}_j} d_{ij}}{n_j + \lambda_\beta}.
\end{equation}
Thus, at the coordinate-update level, every rater contributes equally to the
quality estimate. The exact conditional note-block solution jointly updates
\((\beta_j,\delta_j)\) and therefore includes a coupling term between the
quality and ideology channels; for completeness we give the formal expression in
Appendix~\ref{app:conditional-note-update}.

\subsection{Why Not Simpler Fixes?}\label{sec:simpler-fixes}

Before introducing \QSMF{}, we consider two natural attempts to improve the baseline model.

\paragraph{Higher-dimensional ideology.}
The baseline model captures ideology with a single factor $\gamma_i \delta_j$. A natural idea is to expand to $k$ ideological dimensions, $\sum_{\ell=1}^{k}\gamma_{i\ell}\delta_{j\ell}$, hoping that a richer ideology model will absorb more of the partisan signal and leave a cleaner quality estimate. A large body of evidence from legislative voting and public opinion research shows that political preferences are often well approximated by a single dominant ideological factor: in the US Congress, one dimension of the DW-NOMINATE scaling explains over 83\% of roll-call voting variation, and the second dimension has been negligible since the early 2000s \cite{poole2007,poole2005}. 

Expanding to more factors risks adding variance. The rating matrix is extremely sparse: a typical rater scores fewer than 1\% of all notes, so each rater's and each note's ideology parameters are estimated from a small, non-random slice of the data. The additional factors may absorb noise rather than meaningful ideological structure, inflating the variance of all downstream estimates (including $\hat\beta_j$).

Importantly, the quality-sensitivity mechanism of \QSMF{} is orthogonal to the choice of number of ideology dimensions. In settings where a richer ideology model is warranted, the $\rho_i$ parameterisation applies unchanged. Appendix~\ref{app:k2} replicates both the out-of-sample prediction and rolling-window sample efficiency experiments with $k=2$ ideology dimensions, confirming that the gains from \QSMF{} are not an artefact of the one-dimensional assumption.

\paragraph{Inverse-variance weighting (IVW).}
A second natural idea draws on a classical result in statistics: when observations have unequal variances, the inverse-variance weighted estimator is the maximum likelihood estimator and achieves the minimum variance among all linear unbiased estimators (the Gauss--Markov theorem under heteroskedasticity). Heteroskedastic matrix factorization exploits exactly this principle: by modelling observation-level noise and reweighting accordingly, it extracts more signal per observation, which should translate directly into faster convergence and improved sample efficiency. A natural choice is to model the variance of each rating as a product of a rater-specific factor and a note-specific factor, $\mathrm{Var}(\epsilon_{ij}) = \sigma_i^2 \eta_j^2$, estimate these from the baseline residuals, and reweight by $w_{ij} = 1/(\hat\sigma_i^2 \hat\eta_j^2)$. This gives more influence to raters whose behaviour the model predicts well and to notes with low residual dispersion.

Despite this strong statistical motivation, the approach systematically upweights exactly the wrong raters. A pure partisan whose ratings are fully explained by $\gamma_i \delta_j$ often has small residuals and therefore receives \emph{high} weight. A rater who always rates ``helpful'' or always rates ``not helpful'' is similarly predictable (absorbed by $\hat\alpha_i$) and similarly rewarded. Meanwhile, ideologically moderate or multidimensional raters (those whose true ideology is poorly approximated by a single $\gamma_i$) have inflated residuals and are downweighted, despite often being the most thoughtful evaluators. Decomposing the variance into rater and note factors does not resolve this, because the pathology is in the rater component $\hat\sigma_i^2$: it conflates predictability with genuine quality tracking. 

For the diagnostic in Figure~\ref{fig:ivw-ideology}, we therefore construct a rater-level IVW reputation on the clean baseline fit by taking the inverse of each rater's regularized average squared residual after normalizing each observation by the corresponding note-level residual MSE. Figure~\ref{fig:ivw-ideology} shows that this inverse-variance reputation proxy rises with ideological extremity, so the scheme systematically amplifies ideologically extreme raters. This makes the algorithm even more vulnerable to partisan attacks on proposed notes.

\begin{figure}[t]
\centering
\includegraphics[width=0.5\linewidth]{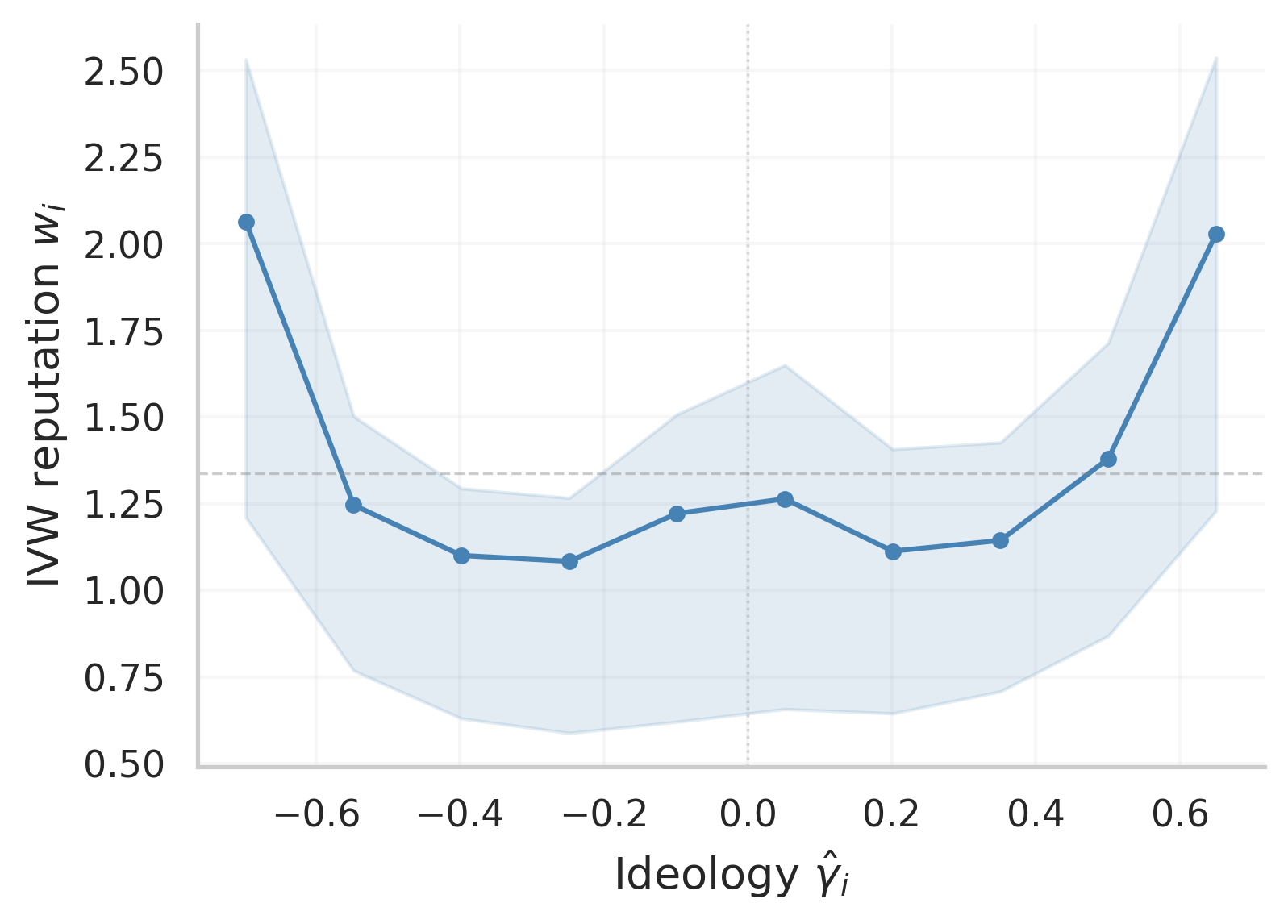}
\caption{Inverse-variance weighting rewards predictability rather than quality tracking. The plotted IVW reputation is the inverse of a regularized, note-normalized residual from the clean baseline fit on the CN data from the 6 months prior the 2024 US presidential elections, and it rises sharply with ideology intensity $|\hat\gamma_i|$.}
\label{fig:ivw-ideology}
\end{figure}

\QSMF{} retains the desirable property of inverse-variance weighting (highly consistent, quality-tracking raters are effectively upweighted). Rather than rewarding predictability through \emph{any} channel, $\hat\rho_i$ is explicitly tied to alignment with the learned note quality values $\hat\beta_j$ on previously rated notes, after ideology has been partialled out. A rater who is predictable because of ideology alone receives low $\hat\rho_i$; a rater who is predictable because they track quality receives high $\hat\rho_i$. The distinction is precisely the one that inverse-variance weighting cannot make.

\section{Quality-Sensitive Matrix Factorization}\label{sec:stage2}

We now discuss \QSMF{}. 
We model observed ratings using the following data-generating process.

\begin{assumption}[Data generation process (DGP)]\label{ass:dgp-homo}
We assume that the ratings can be modeled as:
\begin{equation}\label{eq:dgp-homo}
    r_{ij} = \mu + \alpha_i + \rho_i \beta_j + \gamma_i \delta_j + \epsilon_{ij}, \quad \epsilon_{ij} \overset{\mathrm{iid}}{\sim} \textsc{Normal}(0, \sigma^2).
\end{equation}
with \(\rho_i \ge 0\) for all raters \(i\).
\end{assumption}

Relative to the baseline model in~\eqref{eq:baseline-dgp}, the only change is that the note-quality term $\beta_j$ is replaced by the bilinear term $\rho_i \beta_j$. This preserves the interpretation of $\beta_j$ as latent note quality, but allows raters to differ in how strongly their ratings reflect that quality. Equivalently, the rating decomposes as
\begin{equation}
    r_{ij}
    =
    \underbrace{\mu + \alpha_i}_{\text{rater baseline}}
    +
    \underbrace{\rho_i \beta_j}_{\text{quality channel}}
    +
    \underbrace{\gamma_i \delta_j}_{\text{ideology channel}}
    +
    \underbrace{\epsilon_{ij}}_{\text{noise}}.
\end{equation}

We impose $\rho_i \ge 0$, to ensure identifiability between the two channels.

We fit the model by regularised least squares over the observed entries $\Omega$. As in the baseline CN model, the regularisation terms play both statistical and interpretive roles: they stabilise estimation in a sparse matrix, select a canonical factorisation, and encode sensible defaults.

\begin{equation}\label{eq:objective-homo}
\begin{aligned}
    \mathcal{L}(\theta) &= \frac{1}{2}\sum_{(i,j) \in \Omega} \left(r_{ij} - \mu - \alpha_i - \rho_i \beta_j - \gamma_i \delta_j\right)^2 \\
    &\quad + \frac{\lambda_\rho}{2} \sum_{i \in \mathcal I} (\rho_i - 1)^2
    + \frac{\lambda_\alpha}{2} \sum_{i \in \mathcal I} \alpha_i^2
    + \frac{\lambda_\gamma}{2} \sum_{i \in \mathcal I} \gamma_i^2 
    + \frac{\lambda_\beta}{2} \sum_{j \in \mathcal N} \beta_j^2
    + \frac{\lambda_\delta}{2} \sum_{j \in \mathcal N} \delta_j^2.
\end{aligned}
\end{equation}

The prior center for $\rho_i$ is important. Centering the penalty at $\rho_i = 1$ makes the baseline model the default. 
Recall the de-ideologized rating \(d_{ij}\) from
\eqref{eq:deideologized-intro}.
Holding the ideology coordinate fixed, the quality-coordinate update is
\begin{equation}\label{eq:beta-update-model}
    \hat{\beta}_j = \frac{\sum_{i \in \mathcal{I}_j} \hat{\rho}_i \, d_{ij}}{S_j + \lambda_\beta},
    \qquad
    S_j = \sum_{i \in \mathcal{I}_j} \hat{\rho}_i^2.
\end{equation}
Compared with the baseline update~\eqref{eq:beta-baseline}, this is the key structural change: note quality is no longer the simple average of de-ideologized ratings. Instead, each rater's contribution is weighted by $\hat\rho_i$. The term $S_j$ plays the role of an effective sample size. Two notes may receive the same number of ratings, but the note rated by more quality-sensitive raters is treated as having more usable information.
This scalar formula should be read as the quality-coordinate view of the note
subproblem. The exact conditional note-block solution jointly updates
\((\beta_j,\delta_j)\) and includes an overlap term capturing the empirical
alignment between the quality and ideology channels on the rater pool; see Appendix~\ref{app:conditional-note-update}. 

The ideology update remains structurally similar:
\begin{equation}
    \hat{\delta}_j
    =
    \frac{\sum_{i \in \mathcal{I}_j} \hat{\gamma}_i \, (r_{ij} - \hat{\mu} - \hat{\alpha}_i - \hat{\rho}_i \hat{\beta}_j)}
    {\sum_{i \in \mathcal{I}_j} \hat{\gamma}_i^2 + \lambda_\delta}.
\end{equation}
Thus \QSMF{} does not replace the bridging mechanism; it augments it with a separate quality-sensitivity channel. The ideology term continues to absorb partisan alignment, while the $\rho_i \beta_j$ term governs how much each rater should influence the learned quality score.

We discuss the induced update rule for $\rho_i$ in the next section.

\subsection{Identifiability}\label{sec:ident}
 \QSMF{} is non-convex and contains bilinear terms, so identification requires conventions. Our regularisation and sign constraints jointly select one representative from the usual equivalence class of factorizations.
The product $\rho_i \beta_j$ is invariant under the rescaling
\(
(\rho_i,\beta_j) \mapsto (c\rho_i,\beta_j/c),~ c>0.
\)
 The penalties $\lambda_\rho(\rho_i-1)^2$ and $\lambda_\beta \beta_j^2$ anchor the scale by preferring values of $\rho_i$ near one and values of $\beta_j$ with controlled magnitude. 

With two bilinear channels, one must also distinguish the ``quality'' factor from the ``ideology'' factor. In \QSMF{}, these roles are separated by modeling restrictions. The quality channel has nonnegative rater factors $\rho_i$ and is regularised toward the baseline value $1$; the ideology channel has signed factors $\gamma_i$ and is regularised toward $0$. In addition, we use a two-timescale training procedure, so the ideology factor begins aligned with the standard CN interpretation. Together, these conventions give the two channels distinct meanings and make the learned decomposition stable in practice.

\section{Theoretical and Structural Properties of \QSMF{}}\label{sec:theory}

The goal of this section is to quantify when heterogeneous weighting improves
estimation, interpret the learned parameters, and show
why the model tends to attenuate certain forms of manipulation.

Let
\(
\theta^*
=
\bigl(
\mu^*,\{\alpha_i^*\},\{\rho_i^*\},\{\beta_j^*\},\{\gamma_i^*\},\{\delta_j^*\}
\bigr)
\)
denote the true parameters under the DGP in Assumption~\ref{ass:dgp-homo}. Let
\(\mathcal L(\theta)\) denote the regularized objective in
\eqref{eq:objective-homo}, and let
\(
\hat\theta
=
\bigl(
\hat\mu,\{\hat\alpha_i\},\{\hat\rho_i\},\{\hat\beta_j\},\{\hat\gamma_i\},\{\hat\delta_j\}
\bigr)
\)
denote a fixed point of this objective. We use
\(
\hat d_{ij}
:=
r_{ij}-\hat\mu-\hat\alpha_i-\hat\gamma_i\hat\delta_j
\)
for the fitted de-ideologized rating, and
\(
d^*_{ij}
:=
r_{ij}-\mu^*-\alpha_i^*-\gamma_i^*\delta_j^*
\)
for its oracle counterpart under the DGP.
We now turn to the statistical and structural consequences of that update
structure. First, we quantify the statistical gain from heterogeneous weighting
in an oracle benchmark where the nuisance channels have already been removed.
Second, we characterize the fitted parameter \(\rho_i\), and discuss the structural implications of its update rule
for simple, persistently uninformative attack patterns. Finally, we derive a diagnostic for whether heterogeneous quality sensitivity is
supported by the data.

\subsection{Efficiency Under Heterogeneous \texorpdfstring{$\rho_i$}{rho_i}}
\label{sec:efficiency}

We now isolate the basic statistical benefit of allowing raters to have
different \(\rho_i\) values. Once bias and ideology
have been removed, a rater with larger \(\rho_i\) reacts more strongly to note
quality than a rater with smaller \(\rho_i\). Such a rater therefore carries
more information about \(\beta_j\). If the \(\rho_i\) values truly differ across
raters, then treating all raters equally discards information.
We work in an oracle reduced model in which the
nuisance channels have already been removed. 

Fix a note \(j\). Under the DGP in Assumption~\ref{ass:dgp-homo}, the oracle
de-ideologized rating \(d^*_{ij}\) satisfies
\begin{equation}
\label{eq:oracle-reduced-model}
d^*_{ij}
=
\rho_i^* \beta_j^* + \epsilon_{ij},
\qquad
\mathbb E[\epsilon_{ij}] = 0,
\qquad
\Var(\epsilon_{ij}) = \sigma^2.
\end{equation}
Thus, after the nuisance terms have been removed, the only remaining difference
across raters is the multiplicative factor \(\rho_i^*\). 
Conditional on the rater pool \(\mathcal I_j\), estimating \(\beta_j^*\) from
\eqref{eq:oracle-reduced-model} is simply a one-dimensional regression problem
with heterogeneous \((\rho_i^*)_{i\in\mathcal I_j}\). The
oracle estimator is therefore
\begin{equation}
\label{eq:oracle-weighted-est}
\hat\beta_j^{\rho}
=
\frac{\sum_{i\in\mathcal I_j}\rho_i^* d^*_{ij}}
{\sum_{i\in\mathcal I_j}(\rho_i^*)^2}.
\end{equation}

Assume that
\(
\sum_{i\in\mathcal I_j}(\rho_i^*)^2 > 0,
\)
so that note \(j\) is rated by at least one rater with nonzero \(\rho_i^*\).
Since the model imposes \(\rho_i^* \ge 0\), this implies
\(
\bar\rho_j^*
:=
\frac{1}{n_j}\sum_{i\in\mathcal I_j}\rho_i^*
> 0.
\)
We also assume that, conditional on the rater pool \(\mathcal I_j\), the
errors \(\{\epsilon_{ij} : i\in\mathcal I_j\}\) are independent across raters.

\begin{proposition}[Oracle efficiency under heterogeneous \(\rho_i\)]
\label{prop:efficiency}
Under \eqref{eq:oracle-reduced-model} and cross-rater independence of the
errors \(\epsilon_{ij}\), the estimator
\(\hat\beta_j^{\rho}\) in \eqref{eq:oracle-weighted-est} is unbiased for
\(\beta_j^*\) and has variance \( \Var(\hat\beta_j^{\rho})
= \frac{\sigma^2}{\sum_{i\in\mathcal I_j}(\rho_i^*)^2}. \)

Among estimators that aggregate the same oracle de-ideologized ratings using a
common coefficient on every rater,
\(
\tilde\beta_j = c_j \sum_{i\in\mathcal I_j} d^*_{ij},
\)
the unique unbiased choice is \( \tilde\beta_j^{\mathrm{uni}}
= \frac{1}{n_j\bar\rho_j^*}\sum_{i\in\mathcal I_j} d^*_{ij},  \)
and its variance is \( \Var(\tilde\beta_j^{\mathrm{uni}}) =\frac{\sigma^2}{n_j(\bar\rho_j^*)^2}. \)
\end{proposition}

\begin{proof}
See Appendix~\ref{app:proof-efficiency}.
\end{proof}

The proposition compares two ways of using the same de-ideologized ratings. The
estimator \(\hat\beta_j^\rho\) respects the heterogeneity in the
\(\rho_i^*\) values, whereas \(\tilde\beta_j^{\mathrm{uni}}\) treats every
rater symmetrically. The resulting variance gap is explicit.

\begin{corollary}[Variance inflation from equal weighting]
\label{cor:variance-ratio}
Under the same assumptions,
\begin{equation}
\label{eq:variance-ratio}
\frac{\Var(\tilde\beta_j^{\mathrm{uni}})}
{\Var(\hat\beta_j^{\rho})}
=
\frac{\overline{(\rho^*)^2}_j}{(\bar\rho_j^*)^2}
=
1+\frac{\Var_{\mathcal I_j}(\rho_i^*)}{(\bar\rho_j^*)^2},
\end{equation}
where
\[
\overline{(\rho^*)^2}_j
:=
\frac{1}{n_j}\sum_{i\in\mathcal I_j}(\rho_i^*)^2.
\]
Hence \(\tilde\beta_j^{\mathrm{uni}}\) is strictly less efficient whenever the
\(\rho_i^*\) values are not all identical on the rater pool for note \(j\).
\end{corollary}

Corollary~\ref{cor:variance-ratio} formalizes the price of ignoring
heterogeneity. The ratio in \eqref{eq:variance-ratio} is always at least one,
with equality only when all \(\rho_i^*\) values are the same. 
It is important to interpret this result at the right level. Proposition~\ref{prop:efficiency}
does not claim that the fully estimated \QSMF{} estimator uniformly dominates
every possible alternative in finite samples. Rather, it isolates the basic
statistical gain available in the oracle reduced model: once the nuisance
channels have been removed, heterogeneity in \(\rho_i^*\) is itself a usable
source of information, and any estimator that ignores it leaves efficiency on
the table. The rolling-window experiments in Section~\ref{sec:sample-efficiency} confirm this prediction empirically: \QSMF{} reaches comparable note-quality accuracy with 26--40\% fewer ratings than the baseline.

\subsection{Interpreting \texorpdfstring{$\rho_i$}{rho_i} and the Attenuation of Naive Attacks}
\label{sec:rho-update}

At a fixed point, the \(\rho_i\)-subproblem is
\begin{equation}
\label{eq:rho-subproblem}
\min_{\rho\ge 0}\;
\frac12\sum_{j\in\mathcal N_i}(\hat d_{ij}-\rho\hat\beta_j)^2
+\frac{\lambda_\rho}{2}(\rho - 1)^2,
\end{equation}
where \(\hat d_{ij}\) is the fitted de-ideologized rating defined at the
start of this section.
Thus, conditional on the other fitted parameters, \(\rho_i\) is the coefficient
in a one-dimensional nonnegative ridge regression of de-ideologized ratings on
the fitted quality axis. Define
\[
A_i := \sum_{j\in\mathcal N_i}\hat d_{ij}\hat\beta_j,
\qquad
B_i := \sum_{j\in\mathcal N_i}\hat\beta_j^2,
\qquad
D_i := \sum_{j\in\mathcal N_i}\hat d_{ij}^2.
\]
On solving~\eqref{eq:rho-subproblem}, we get
\[
\hat\rho_i
=
\frac{(A_i + \lambda_\rho)_+}{B_i+\lambda_\rho},
\]
where \((x)_+ := \max\{x,0\}\); see
Appendix~\ref{app:rho-update-derivation} for the derivation from
\eqref{eq:rho-subproblem}. This gives the main interpretation of
\(\hat\rho_i\). The alignment term \(A_i\) is large when rater \(i\)'s
de-ideologized ratings move with the fitted quality, while the
regularizer \(\lambda_\rho\) pulls the estimate toward the baseline value
\(\rho_i=1\). 

The parameter \(\hat\rho_i\) measures the \emph{magnitude} of quality tracking:
how strongly rater~\(i\)'s de-ideologized ratings respond to note quality. A
complementary question is what \emph{fraction} of the cross-note variation in
those de-ideologized ratings is explained by the quality signal. Two raters with
identical \(\hat\rho_i\) can differ substantially in this fraction if one has
noisier de-ideologized residuals. Define the fitted quality-channel \(R^2\) as
\begin{equation}
\label{eq:partial-r2}
R^{2,\mathrm{qual}}_i
:=
1-
\frac{\sum_{j\in\mathcal N_i}(\hat d_{ij}-\hat\rho_i\hat\beta_j)^2}{D_i},
\end{equation}
with the convention \(R^{2,\mathrm{qual}}_i=0\) when \(D_i=0\).

\paragraph{Unregularized case.}
The connection between \(\hat\rho_i\) and \(R^{2,\mathrm{qual}}_i\) is most
transparent without regularization. Setting \(\lambda_\rho=0\), the
\(\rho\)-update reduces to
\(\hat\rho_i = (A_i)_+/B_i\). When \(A_i>0\) (so that
\(\hat\rho_i>0\)), the normal equation gives \(A_i = \hat\rho_i B_i\).
This yields
\begin{equation}\label{eq:r2-decomp}
R^{2,\mathrm{qual}}_i
=
\frac{2\hat\rho_i A_i - \hat\rho_i^2 B_i}{D_i}
=
\frac{2\hat\rho_i^2 B_i - \hat\rho_i^2 B_i}{D_i}
=
\hat\rho_i^{\,2}\;\frac{B_i}{D_i};
\end{equation}
when \(A_i \le 0\), \(\hat\rho_i=0\) and
\(R^{2,\mathrm{qual}}_i=0\). The takeaway is direct: a higher \(\hat\rho_i\)
means that a larger share of the rater's de-ideologized variation is genuine
quality signal rather than noise.

\paragraph{Population limit.}
Under the oracle reduced model~\eqref{eq:oracle-reduced-model}, as the number
of notes rated by rater~\(i\) grows, the law of large numbers gives
\(B_i^*/|\mathcal N_i| \to \E[\beta_j^{*2}]\)
and
\(D_i^*/|\mathcal N_i| \to (\rho_i^*)^2\E[\beta_j^{*2}]+\sigma^2\),
where \(B_i^*=\sum_j\beta_j^{*2}\),
\(D_i^*=\sum_j d_{ij}^{*2}\), and the expectations are over the distribution
of note qualities. The oracle quality-channel \(R^2\) therefore converges to
\begin{equation}\label{eq:pop-r2}
R^{2,\mathrm{pop}}_i
\;=\;
\frac{(\rho_i^*)^2\,\E[\beta_j^{*2}]}
     {(\rho_i^*)^2\,\E[\beta_j^{*2}]+\sigma^2},
\end{equation}
which is strictly increasing in \(\rho_i^*\). In the population, quality
sensitivity and quality-channel \(R^2\) are linked by a one-to-one monotone
mapping: raters with higher \(\rho_i^*\) explain a larger share of their
de-ideologized variation through quality.

\paragraph{Interpretation as attack attenuation.}
The same update formula also clarifies why \QSMF{} attenuates naive attacks.
In the baseline CN model, every rating enters the quality channel with equal,
fixed influence. In \QSMF{}, influence is mediated by \(\hat\rho_i\), which is
continuously adjusted according to observed alignment with the fitted quality
signal. 
When a rater's behavior is uninformative, \(A_i\) is near zero, so
\(
\hat\rho_i \approx \frac{\lambda_\rho}{B_i+\lambda_\rho}.
\)
Observe that \(\hat\rho_i\) is estimated from the full
rater history \(\mathcal N_i\) through the aggregate quantities \(A_i\) and
\(B_i\), rather than just from newly rated notes. In practice, this means
that the update is typically dominated by notes whose quality estimates
\(\hat\beta_j\) have already stabilized because they have accumulated many
ratings. As a result, \(\hat\rho_i\) does not respond too sharply to noisy
early \(\hat\beta_j\) values on newly rated notes; instead, it moves mainly in
response to persistent patterns of alignment or misalignment with the fitted
quality signal.

Relative to the baseline, which provides no feedback mechanism, \QSMF{} imposes
an ex post cost on low quality ratings. The attenuation grows with the severity
and persistence of the behavior, but it should be read as a structural result
for naive attacks rather than a full strategyproofness guarantee. The corruption experiments in Section~\ref{sec:individual-attacks} confirm that attackers (partisans, coin-flippers, and constant raters) receive sharply depressed \(\hat\rho_i\) values, and the coordinated attack experiments in Section~\ref{sec:targeted-attacks} show that this attenuation translates into smaller displacement of attacked note scores.

\subsection{A Score-Based Diagnostic for Heterogeneous Quality Sensitivity}
\label{sec:score-test}

Before using \QSMF{}, it is useful to ask whether the homogeneous-sensitivity
assumption of the baseline Community Notes model is empirically reasonable. If
\(\rho_i\approx 1\) for all raters, then introducing heterogeneous quality
sensitivity adds parameters without capturing real structure. The baseline model
implicitly assumes that, after removing each rater's bias and ideological
component, every rater's de-ideologized ratings respond to note quality with the
same slope. If this slope varies systematically across raters, the
homogeneous-sensitivity assumption is violated.

\textbf{The slope statistic.}
Using the fitted baseline parameters, define rater~\(i\)'s baseline
de-ideologized rating
\(
d_{ij}^{\mathrm{base}}
=
r_{ij}
-\hat\mu^{\mathrm{base}}
-\hat\alpha_i^{\mathrm{base}}
-\hat\gamma_i^{\mathrm{base}}\hat\delta_j^{\mathrm{base}}
\)
and the OLS slope of these de-ideologized ratings on the baseline quality
estimates:
\[
s_i
=
\frac{\sum_{j\in\mathcal N_i}
      d_{ij}^{\mathrm{base}}\,\hat\beta_j^{\mathrm{base}}}
     {\sum_{j\in\mathcal N_i}(\hat\beta_j^{\mathrm{base}})^2}.
\]
Under the homogeneous-sensitivity null (\(\rho_i=1\) for all~\(i\)), this
slope should be near~\(1\) for every rater.

\textbf{Connection to the Rao score.}
The baseline residual is
\(
e_{ij}^{\mathrm{base}}
= r_{ij}
  -\hat\mu^{\mathrm{base}}
  -\hat\alpha_i^{\mathrm{base}}
  -\hat\beta_j^{\mathrm{base}}
  -\hat\gamma_i^{\mathrm{base}}\hat\delta_j^{\mathrm{base}},
\)
so the de-ideologized rating decomposes as
\(d_{ij}^{\mathrm{base}} = e_{ij}^{\mathrm{base}} + \hat\beta_j^{\mathrm{base}}\).
Substituting into \(s_i\) and writing
\(B_i^{\mathrm{base}}=\sum_j(\hat\beta_j^{\mathrm{base}})^2\) gives
\[
s_i
=
\frac{\sum_j(e_{ij}^{\mathrm{base}}+\hat\beta_j^{\mathrm{base}})
      \hat\beta_j^{\mathrm{base}}}{B_i^{\mathrm{base}}}
=
1 + \frac{U_i}{B_i^{\mathrm{base}}},
\qquad
U_i := \sum_{j\in\mathcal N_i}\hat\beta_j^{\mathrm{base}}\,
       e_{ij}^{\mathrm{base}}.
\]
The quantity \(U_i\) is the Rao score component
\cite{rao1948large} for the parameter \(\rho_i\) evaluated at the restricted
estimate \(\rho_i=1\): it is the derivative of the log-likelihood of the
unrestricted model with respect to~\(\rho_i\), evaluated at the baseline MLE
(up to a \(1/\sigma^2\) scaling; see Appendix~\ref{app:score-test}). Under the
null, \(\E[U_i]=0\), so \(s_i\) is centered at~\(1\). Systematic deviations of
\(s_i\) from~\(1\) therefore provide evidence of heterogeneous quality
sensitivity.

\textbf{Connection to \QSMF{}.}
The statistic \(s_i\) is the baseline analogue of the quantity that \QSMF{}
estimates as \(\hat\rho_i\). Ignoring regularization and the nonnegativity
constraint in~\eqref{eq:rho-subproblem}, the fitted \(\rho_i\) is exactly the
slope of de-ideologized ratings on the current quality axis. The diagnostic
\(s_i\) computes the same slope using only baseline parameters.
Section~\ref{sec:dgp-validation} applies this diagnostic to real data,
confirming that \(s_i\) exhibits substantial dispersion and strong split-half
reliability, and that the model-estimated \(\hat\rho_i\) correlates highly with
the slope~\(s_i\).

\section{Experimental Setup}
\subsection{Data}\label{sec:data}

We use the publicly released CN dataset from X, which provides all ratings, notes, and note status histories.
Our sample covers all ratings from May--October 2024 (44.99M ratings, 412,381 raters, and 365,431 notes after filtering to raters and notes with more than 10 ratings each). This corresponds to the six months before the 2024 U.S.\ presidential election. We only consider the notes rated by the `core' model on X, i.e., notes on political tweets.

\subsection{Training Procedure}\label{sec:training}

The objective~\eqref{eq:objective-homo} is non-convex due to the bilinear terms $\rho_i \beta_j$ and $\gamma_i \delta_j$. We optimize via alternating minimization with timescale separation \cite{borkar1997}:

\begin{enumerate}[leftmargin=*]
    \item \textbf{Fast step.} Fix $\rho$, optimize $(\mu, \alpha, \beta, \gamma, \delta)$ via Adam \cite{kingma2015} with gradient clipping (max norm 1.0) until convergence or an epoch ceiling.
    \item \textbf{Slow step.} Fix all other parameters, solve $\rho$ in closed form:
    \begin{equation}\label{eq:rho-closed-form}
        \hat{\rho}_i = \max\!\left(0,\; \frac{\sum_{j \in \mathcal N_i} \hat{\beta}_j\, d_{ij} + \lambda_{\mathrm{eff}}}{\sum_{j \in \mathcal N_i} \hat{\beta}_j^2 + \lambda_{\mathrm{eff}}}\right),
    \end{equation}
    where $d_{ij}$ denotes the current de-ideologized residual from
    \eqref{eq:deideologized-intro}, $N = |\Omega|$ is the number of observed
    ratings, and $M = |\mathcal I|$ is the number of raters.  $\lambda_{\mathrm{eff}} = \lambda_\rho \cdot N / M.$ 
\end{enumerate}

These two steps alternate for $R$ rounds. 
For the baseline we only have the fast step until convergence.

\paragraph{Hyperparameters.} We use $R = 5$ rounds since it is sufficient for $\rho$ estimates to stabilize. In all runs we use convergence detection within the fast step, terminating when the loss changes by less than $10^{-6}$ for 20 consecutive epochs. Epoch ceiling is 1000.  We use learning rate $0.1$. Regularization is $\lambda_\alpha = \lambda_\beta = \lambda_\gamma = \lambda_\delta = \lambda_\rho = 0.02$. 

\paragraph{Post-Estimation Normalization}
After training, we apply the identifiability normalization $\hat{\rho}_i \leftarrow \hat{\rho}_i / \bar{\rho}$, $\hat{\beta}_j \leftarrow \hat{\beta}_j \cdot \bar{\rho}$, where $\bar{\rho} = M^{-1}\sum_i \hat{\rho}_i$. This anchors the mean quality sensitivity at 1, ensuring $\hat{\beta}_j$ is on the same scale as the baseline $\hat{\beta}_j$. Wherever applicable, we report the mean squared error in estimates after z-scoring (setting mean to 0 and variance to 1) to get a fair comparison, since finally only the relative values of note quality estimates are needed for classification into helpful or not helpful.

\subsection{What the Full-Data Fit Learns}\label{sec:learned-params}

Before turning to formal validation tests, it is useful to inspect the learned latent variables themselves. Figure~\ref{fig:param-distributions} shows the full-data distributions of the four main quantities.  The 10th, 5th, and 1st percentiles of $\rho$ are 0.786, 0.696, and 0.477, so those raters receive only $0.79\times$, $0.70\times$, and $0.48\times$ of the baseline influence, respectively. At the same time, the median rater remains essentially at baseline. This lower tail is also tunable through $\lambda_\rho$: decreasing $\lambda_\rho$ would widen the $\hat{\rho}$ distribution further and impose stronger attenuation on noisy or strategic raters, while increasing $\lambda_\rho$ would pull the fit back toward the homogeneous baseline. Our adopted value $\lambda_\rho = 0.02$ is intentionally conservative in this sense.

\begin{center}
\small
\begin{tabular}{lccccccc}
\toprule
Percentile & 1 & 5 & 10 & 50 & 90 & 95 & 99 \\
\midrule
Relative weight vs.\ baseline & $0.48\times$ & $0.70\times$ & $0.79\times$ & $1.01\times$ & $1.21\times$ & $1.28\times$ & $1.45\times$ \\
\bottomrule
\end{tabular}
\end{center}

\begin{figure}[t]
\centering
\includegraphics[width=0.8\linewidth]{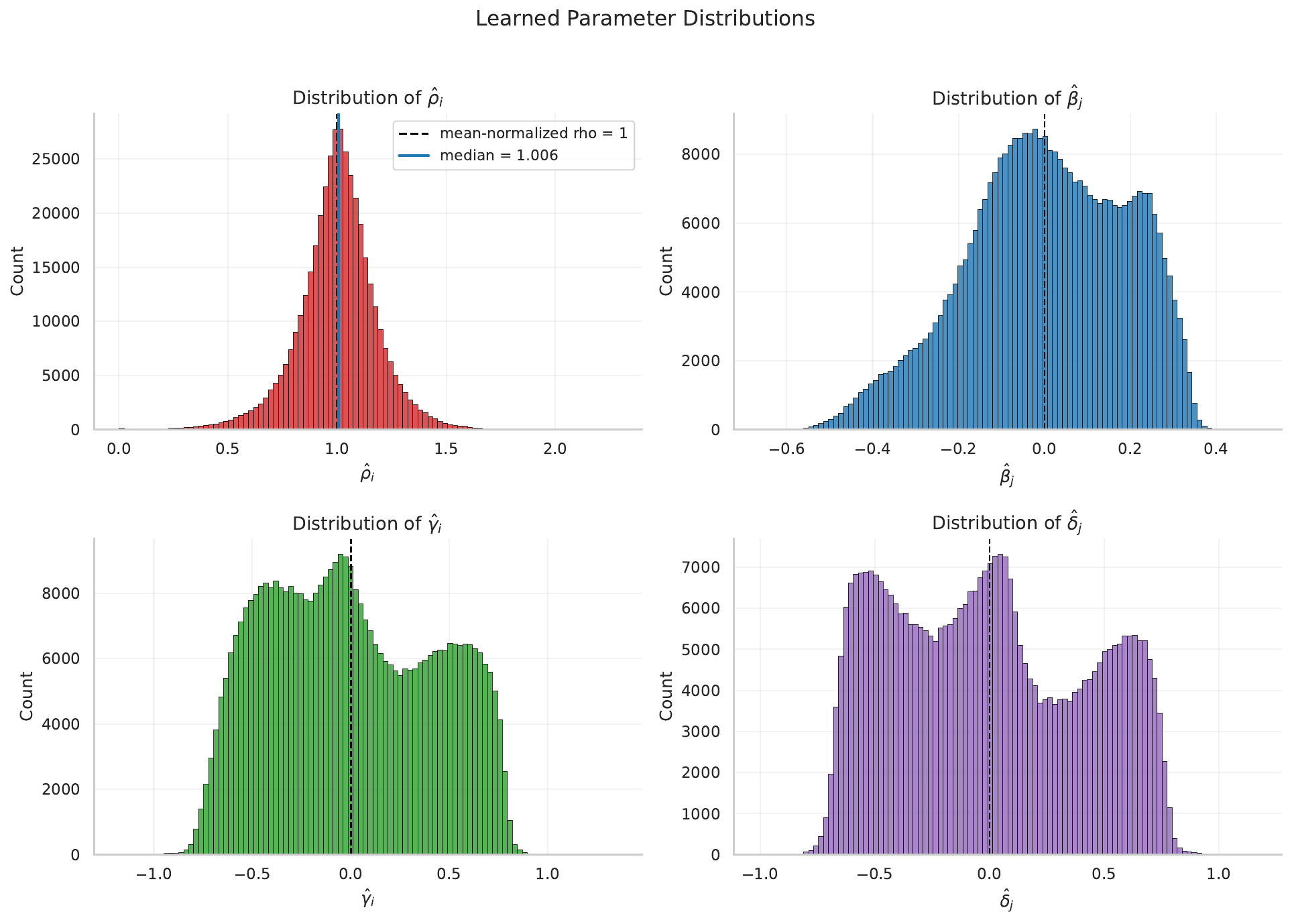}
\caption{Latent-parameter distributions from the full-data \QSMF{} fit. The key object is the lower tail of $\hat{\rho}_i$; the signs of $\hat{\gamma}_i$ and $\hat{\delta}_j$ are arbitrary up to a global flip.}
\label{fig:param-distributions}
\end{figure}
\begin{figure}[h]
\centering
\includegraphics[width=0.5\linewidth]{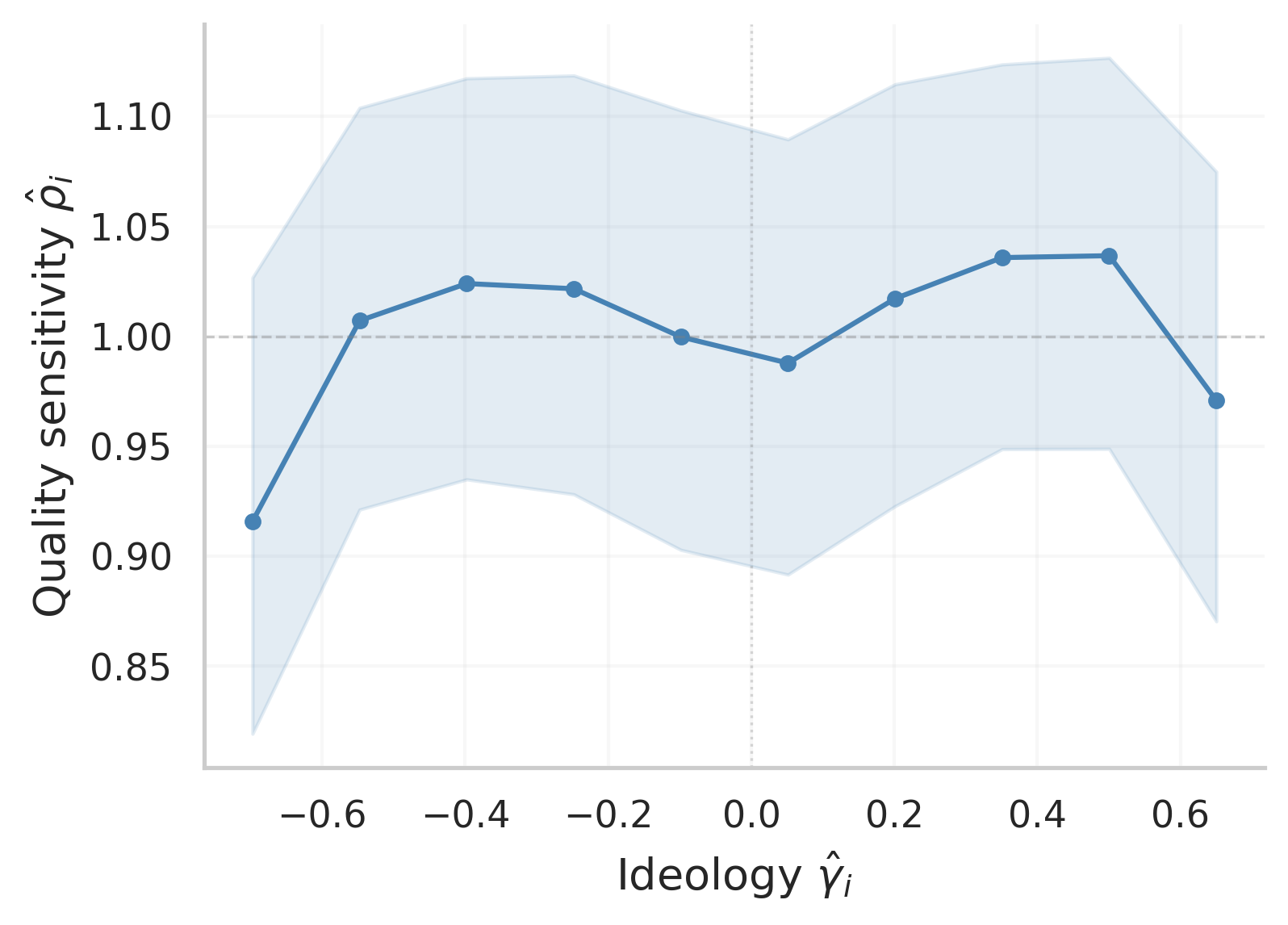}
\caption{Fitted quality sensitivity $\hat\rho_i$ against ideology intensity $|\hat\gamma_i|$. The weak Spearman correlation ($r=-0.075$) shows that $\hat\rho_i$ is not simply an inverse measure of ideological intensity. The shaded region is the 95\% confidence interval. The within-bucket variation is much higher than the variation across ideology buckets.} 
\label{fig:rho-gamma-scatter}
\end{figure}

We also want to check that $\hat{\rho}_i$ is not penalizing specific ideologies or ideological intensities. The relationship with ideology intensity is mild: the Spearman correlation between $\hat{\rho}_i$ and $|\hat{\gamma}_i|$ is $-0.075$. Across the two signs of the fitted ideology axis, the mean $\hat{\rho}$ values are 0.976 and 1.009, a difference of only +0.033. In other words, there is no large sign asymmetry in the fitted ideological coordinate, even though the most ideologically intense tail is attenuated somewhat more strongly on average. Figure~\ref{fig:rho-gamma-scatter} shows the joint distribution of $(\hat\rho_i, \hat\gamma_i)$.

\section{Data Generation Process (DGP) Validation}\label{sec:dgp-validation}

Before evaluating performance, we verify the core premise: raters differ in quality sensitivity, and these differences are stable rather than pure estimation noise. These tests operationalise the diagnostic developed in Section~\ref{sec:score-test}: the model-free slope \(s_i\) is the baseline analogue of \(\hat\rho_i\), and systematic deviations of \(s_i\) from one provide evidence against the homogeneous-sensitivity assumption.

\paragraph{Test 1: Split-Half Slope Consistency}

For each rater $i$ with at least 20 ratings, we randomly split their observations into halves A and B. In each half, we compute the OLS slope $s_i^{(\cdot)} = \Cov(d_{ij}, \hat{\beta}_j) / \Var(\hat{\beta}_j)$, which is a model-free estimator of $\rho_i$ under the DGP~\eqref{eq:dgp-homo}. If quality sensitivity is a stable rater trait, the slopes should correlate across halves; under homogeneous $\rho_i \equiv 1$, the expected correlation is zero.

The split-half Pearson correlation  grows sharply with activity (Table~\ref{tab:dgp-activity}), indicating more precise slope estimates. It is 0.89 for raters with 500 ratings or more. Over the entire population, it is $r = 0.408$ (95\% bootstrap CI: $[0.403, 0.412]$). After partialling out signed ideology and ideology intensity  ($\gamma_i$, $|\gamma_i|$, $|\gamma_i|^2$) for computing correlation via a linear regression fit, the split-half correlation remains 0.391 (95\% CI: $[0.386, 0.397]$), so most of the signal survives ideology controls.
\begin{table}[h]
\begin{center}
\begin{tabular}{lccc}
\toprule
Activity bin & $n$ raters & Split-half $r$ & IQR$(s_i)$ \\
\midrule
20--50 & 137,551 & 0.327 & 0.676 \\
50--100 & 76,434 & 0.499 & 0.481 \\
100--200 & 49,024 & 0.644 & 0.409 \\
200--500 & 32,387 & 0.777 & 0.367 \\
500--$\infty$ & 14,935 & 0.890 & 0.339 \\
\bottomrule
\end{tabular}
\end{center}
\caption{Split-half reliability of the model-free slope by activity bin. Reliability rises sharply with rater activity, indicating a stable latent trait rather than pure noise.}
\label{tab:dgp-activity}
\end{table}

\paragraph{Test 2: Slope Dispersion (Permutation Test)}
We compute the interquartile range $D = \mathrm{IQR}(s_i)$ of rater slopes and compare it to a permutation null that shuffles rater identities on ratings. Under homogeneous $\rho_i$, the observed dispersion should not exceed the permutation distribution.
The empirical dispersion is $D = 0.499$. The permutation null has mean 0.245 (s.d.\ 0.001), so the observed dispersion is 2.04 times the null expectation; none of the 100 permutations exceed the empirical value ($p < 0.01$).

\paragraph{Test 3: Convergent Validity: $\hat{\rho}_i$ vs.\ Model-Free Slope}

If \QSMF{} is working correctly, the estimated $\hat{\rho}_i$ should correlate strongly with the model-free slope $s_i$, despite being estimated through a different procedure (bilinear MF versus per-rater OLS).

Pearson $r(\hat{\rho}_i, s_i) = 0.717$ (95\% CI: $[0.714, 0.719]$), with Spearman $r = 0.818$. After partialling out ideology ($\gamma_i$, $|\gamma_i|$, $|\gamma_i|^2$), the partial correlation remains 0.708 (95\% CI: $[0.705, 0.711]$). The high partial correlation confirms that $\hat{\rho}_i$ captures quality sensitivity beyond what ideology explains.

\subsection{Out-of-Sample Prediction}\label{sec:oos}

We perform held-out prediction under random masking -- the standard out-of-sample check for matrix factorization models. For each of 10 seeds, we hold out a randomly chosen 20\% subset of ratings, train both models on the remaining 80\%, and evaluate prediction only on note--rater pairs that are represented in the training set. 
Averaged across the 10 seeds, baseline MSE is $0.082 \pm 0.002$, whereas \QSMF{} achieves $0.079 \pm 0.002$, where the $\pm$ values denote standard deviations across seeds. This is a $2.618\% \pm 0.091\%$ reduction in MSE, with improvements ranging from 2.504\% to 2.769\%; \QSMF{} wins in all 10 seeds.

The goal of the system is not to predict future ratings, and simpler solutions such as hyperparameter tuning, increasing the ideology dimension, or heteroskedastic matrix factorization may also improve the prediction accuracy. However, a stable out-of-sample prediction accuracy is reassuring.

\section{Sample Efficiency}\label{sec:sample-efficiency}

This section evaluates the sample-efficiency desideratum introduced in Section~\ref{sec:desiderata}: how quickly the model identifies note quality as ratings arrive. This section tests the efficiency prediction of Corollary~\ref{cor:variance-ratio}: if raters differ in quality sensitivity, heterogeneous weighting should reduce the variance of $\hat\beta_j$ relative to equal-weight aggregation. We implement a rolling-window design with two complementary estimators: chronological first-$k$ ratings and bootstrap subsamples of size $k$.

\begin{figure}[!t]
\centering
\includegraphics[width=\linewidth]{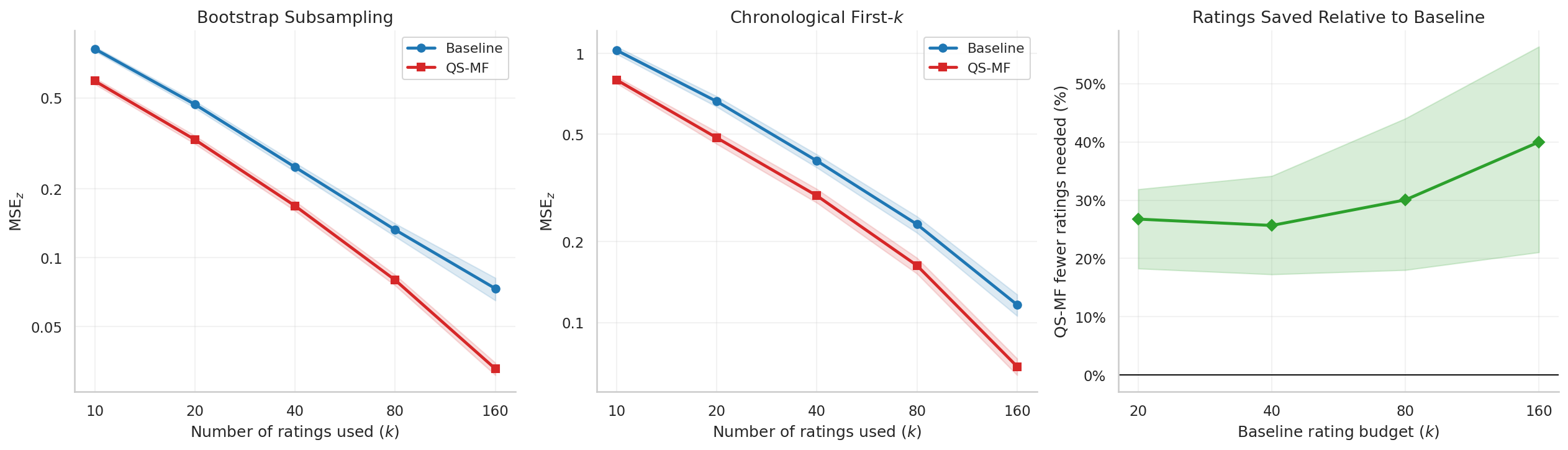}
\caption{Rolling-window sample efficiency across October 2024 for notes first observed in each two-day window and eventually receiving at least 200 ratings. Left: bootstrap subsampling. Middle: chronological first-$k$ estimation. Right: baseline-anchored ratings saved in the bootstrap design. The band shows the range across the 14 evaluation windows after averaging within each window.}
\label{fig:rolling-window}
\end{figure}

\subsection{Rolling-Window Evaluation}\label{sec:rolling-window}

 We define 16 cutoff dates at two-day intervals from October 1 to October 31, 2024. For each cutoff we retrain both models on all ratings strictly before that date. For each evaluation window:
\begin{enumerate}[leftmargin=*]
    \item Freeze rater parameters from the model trained before the window.
    \item Select notes whose first observed rating falls in that two-day window and that eventually receive at least 200 ratings.
    \item Estimate $\hat{\beta}_j^{(k)}$ from either the first $k$ chronological ratings (temporal design) or 100 bootstrap subsamples of size $k$ drawn from the eventual sequence (bootstrap design), with $k \in \{10,20,40,80,160\}$.
    \item Compare to the model (baseline compared with its future version, and QSMF compared with its own) trained on the next cutoff, i.e.\ two days ahead from the end of the window, when most of the ratings for the note have been recorded. For the purpose of this experiment, this future estimate serves as a proxy of a ground truth, denoted as $\beta_j^{\mathrm{gt}}$.
\end{enumerate}

\paragraph{Why freezing rater parameters is necessary.} Re-solving the full MF or \QSMF{} objective separately for every window, note, bootstrap sample, and timestamp would be computationally infeasible at this scale. We therefore freeze the recent rater-side parameters from the most recent pre-window model and recompute only the note-level estimates quickly under that fixed rater state. 

\paragraph{Metric.} We report $\mathrm{MSE}_z(\hat{\beta}_j^{(k)}, \beta_j^{\mathrm{gt}})$ as a function of $k$, where $\mathrm{MSE}_z$ is the mean squared error computed after z-scoring each vector. We also report \emph{ratings saved} in the bootstrap design. For a baseline rating budget $k_{\mathrm{BL}}$, let $k_{\mathrm{QS,eq}}$ be the interpolated number of \QSMF{} ratings needed to match the baseline's error at $k_{\mathrm{BL}}$. We then define
\[
\mathrm{Saved}(k_{\mathrm{BL}})=\left(1-\frac{k_{\mathrm{QS,eq}}}{k_{\mathrm{BL}}}\right)\times 100\%.
\]
This is the operational quantity of interest: if the baseline gets a given accuracy with $k$ ratings, how many fewer ratings would \QSMF{} have needed to get there?

Table entries are mean $\pm$ standard deviation across the 14 window-level aggregates. The ``\% lower'' columns report the proportional reduction in mean MSE relative to the baseline.
\begin{center}
{\small
\setlength{\tabcolsep}{4pt}
\begin{tabular}{rccccccc}
\toprule
& \multicolumn{3}{c}{Bootstrap} & \multicolumn{3}{c}{Temporal} & Ratings saved \\
$k$ & Baseline & \QSMF{} & \% lower & Baseline & \QSMF{} & \% lower & (bootstrap) \\
\midrule
10  & $0.818 \pm 0.036$ & $0.592 \pm 0.035$ & $27.6\%$ & $1.027 \pm 0.062$ & $0.797 \pm 0.043$ & $22.4\%$ & --- \\
20  & $0.468 \pm 0.031$ & $0.328 \pm 0.026$ & $29.9\%$ & $0.663 \pm 0.056$ & $0.486 \pm 0.048$ & $26.7\%$ & $26.7\%$ \\
40  & $0.249 \pm 0.022$ & $0.169 \pm 0.016$ & $32.1\%$ & $0.399 \pm 0.043$ & $0.296 \pm 0.033$ & $25.8\%$ & $25.7\%$ \\
80  & $0.132 \pm 0.016$ & $0.080 \pm 0.008$ & $39.4\%$ & $0.232 \pm 0.031$ & $0.163 \pm 0.021$ & $29.7\%$ & $30.1\%$ \\
160 & $0.073 \pm 0.016$ & $0.033 \pm 0.004$ & $54.8\%$ & $0.116 \pm 0.020$ & $0.069 \pm 0.009$ & $40.5\%$ & $40.0\%$ \\
\bottomrule
\end{tabular}
}
\end{center}

\paragraph{Main findings.} The \QSMF{} error curve lies below the baseline curve at every $k$ in both designs. In the bootstrap design, the mean MSE reduction grows from 27.6\% at $k = 10$ to 54.8\% at $k = 160$. In the chronological first-$k$ design, the corresponding reduction grows from 22.4\% to 40.5\%. A practical summary is given by ratings saved: relative to the baseline, \QSMF{} needs about 27\% fewer ratings at $k = 20$, 26\% fewer at $k = 40$, 30\% fewer at $k = 80$, and 40\% fewer at $k = 160$ to reach the same bootstrap accuracy.

\paragraph{Caveat.} This experiment compares each model against its own future estimate, using the later fit as a proxy for ground truth. As a result, it is informative about sample efficiency conditional on the model class, but it would not detect a failure mode in which a model is consistently wrong and simply converges quickly to its own biased target. We address that limitation in Section~\ref{sec:synthetic}, where ground-truth note quality is known and we can test directly whether \QSMF{} recovers it better than the baseline.

\section{Attack Robustness}\label{sec:attacks}

We evaluate attack robustness with two semi-synthetic experiments on the real data. These experiments test the structural attenuation mechanism described in Section~\ref{sec:rho-update}: partisan bloc voters, random voters, and constant raters all produce de-ideologized residuals with weak alignment to the quality axis, so the model predicts that their $\hat\rho_i$ should fall well below one. In the first experiment, raters make uncoordinated attacks. This experiment asks whether compromised raters receive depressed $\hat{\rho}$ values. Another important experiment is coordinated suppression of notes, where the quantity of interest is displacement on the attacked notes themselves. 

\subsection{Untargeted Corruption}\label{sec:individual-attacks}

We restrict attention to active raters with at least 50 historical ratings (172,780 raters in total). For each $K \in \{1{,}000;\; 5{,}000;\; 20{,}000;\; 50{,}000;\; 100{,}000\}$ and each of 10 seeds, we sample $K$ raters without replacement and corrupt their ratings: one third are assigned to be partisan raters who vote according to ideology alignment, one third are random voters, and one third are lazy raters who always vote helpful. We then train both models on the attacked data. We report the AUC of $\hat{\rho}$ for separating honest raters from attackers, together with the mean $\hat{\rho}$ of each group.

\begin{center}
\begin{tabular}{rccc}
\toprule
$K$ & AUC & Attacker $\hat{\rho}$ & Honest $\hat{\rho}$ \\
\midrule
1,000   & 0.960 & 0.475 & 1.001 \\
5,000   & 0.959 & 0.480 & 1.006 \\
20,000  & 0.962 & 0.495 & 1.026 \\
50,000  & 0.965 & 0.530 & 1.065 \\
100,000 & 0.968 & 0.615 & 1.123 \\
\bottomrule
\end{tabular}
\end{center}

 Across all five attack sizes and all 50 attacked datasets, AUC stays between 0.951 and 0.975, while the mean attacker $\hat{\rho}$ ranges from 0.460 to 0.634 and remains well below the honest-rater mean at every $K$. This result shows that the $\hat{\rho}$ channel is sorting raters in the intended direction. 

\subsection{Coordinated Attacks: $\beta$ Protection}\label{sec:targeted-attacks}

The most direct robustness question is whether \QSMF{} protects note scores from coordinated suppression. When adversaries try to move specific notes, do those attacked notes move less under \QSMF{} than under the baseline?

Motivated by prior work showing that even a small minority of bad raters can strategically suppress targeted helpful notes \cite{truong2025community}, we study a semi-synthetic targeted-suppression design on the real data. We form 100 attacking groups. In each group, one sign of the fitted ideology coordinate is chosen at random. The group samples $K_{\mathrm{pg}}$ ideologically aligned raters from that sign class, requiring at least 50 ratings and $|\hat{\gamma}_i| > 0.3$. This therefore has raters with a well-established $\rho$ from actual ratings. It then targets $K_{\mathrm{pg}}$ high-$\hat{\beta}$ notes from the opposite-sign class, restricted to notes in the top quartile of that class's note-quality distribution and with at least 200 existing ratings. All attacker--target pairs are set to rating 0, and any missing attacker--target ratings are injected explicitly. Given the sparsity of the rating matrix, over 99\% ratings are injections rather than inversions. We sweep $K_{\mathrm{pg}} \in \{10,20,40,80\}$ across 10 seeds. 

For the set of attacked target notes $T$, we compute displacement after z-scoring as
\[
\mathrm{Disp}_z = \frac{1}{|T|}\sum_{j \in T}\bigl[z(\hat{\beta}^{\mathrm{attacked}}_j) - z(\hat{\beta}^{\mathrm{clean}}_j)\bigr]^2,
\]
using each model's own clean estimate as the reference. We define protection as
\[
\mathrm{Protection} = \left(1 - \frac{\mathrm{Disp}_{z,\mathrm{QS}}}{\mathrm{Disp}_{z,\mathrm{BL}}}\right)\times 100\%.
\]

\begin{center}
\begin{tabular}{rrrrrrr}
\toprule
$K_{\mathrm{pg}}$ & \# atk & \# tgt & \# inj & BL disp$_z$ & QS disp$_z$ & Protection \\
\midrule
10 & 996  & 965  & 9,989   & 0.089 & 0.082 & $+8.2\%$  \\
20 & 1,984 & 1,854 & 39,958  & 0.146 & 0.104 & $+28.9\%$ \\
40 & 3,938 & 3,429 & 159,831 & 0.235 & 0.144 & $+38.8\%$ \\
80 & 7,756 & 5,908 & 639,325 & 0.329 & 0.226 & $+31.4\%$ \\
\bottomrule
\end{tabular}
\end{center}

\begin{figure}[h]
\centering
\includegraphics[width=0.8\linewidth]{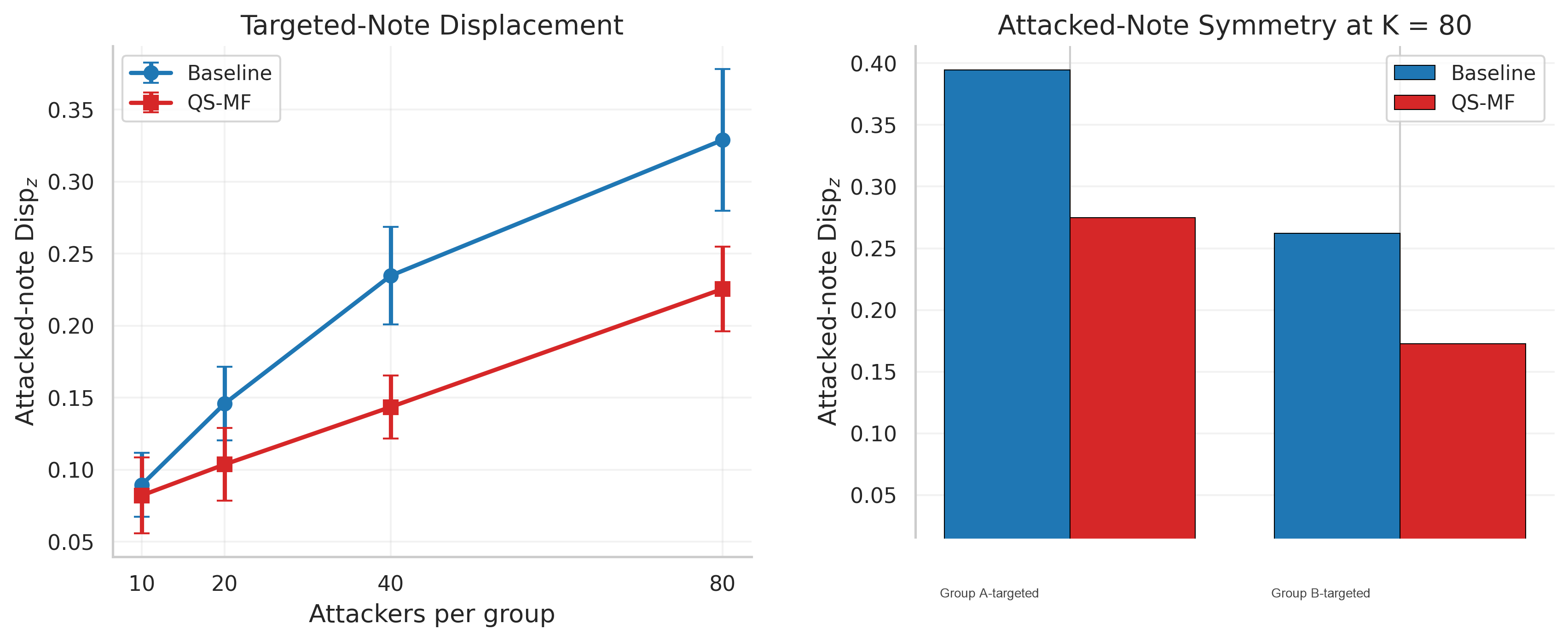}
\caption{Coordinated-attack robustness on attacked notes. Left: attacked-note displacement under the baseline and \QSMF{} as attack size grows; points are mean over seeds and error bars are $\pm 1.96\times$ standard error across 10 seeds. Right: group means for the two sign-defined target sets at the largest attack size.}
\label{fig:attack-robustness}
\end{figure}

Across all attack sizes, \QSMF{} yields smaller attacked-note displacement than
the baseline, with protection ranging from 8.2\% to 38.8\%. The gain is
largest at intermediate attack sizes, but remains positive even in the most
severe setting we consider. Figure~\ref{fig:attack-robustness} also suggests
that this improvement is not confined to one side of the ideological split: at
the largest attack size, both sign-defined target sets move less under \QSMF{}
than under the baseline. The attenuation mechanism therefore does not
eliminate coordinated suppression, but it does materially reduce how far
attacked note scores are displaced. An interesting direction for future work is to study
time-decayed reputations, so that recent behavior receives greater weight when
updating \(\rho_i\).

\section{Synthetic Data: Ground-Truth $\beta$ Recovery}\label{sec:synthetic}

\begin{figure}[t]
\centering
\includegraphics[width=\linewidth]{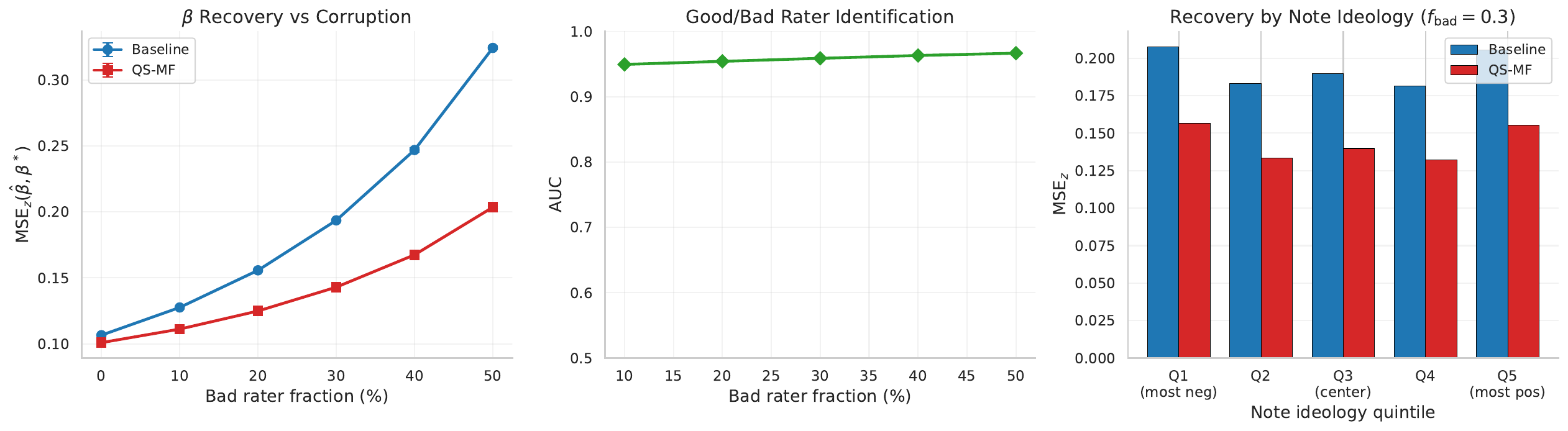}
\caption{Synthetic recovery with known ground truth. Left: $z$-scored note-quality MSE as the bad-rater fraction increases. Middle: AUC for separating good from bad raters using $\hat{\rho}$. Right: recovery error by note-ideology quintile at $f_{\mathrm{bad}}=0.3$.}
\label{fig:synthetic-recovery}
\end{figure}

This section covers three of the desiderata in
Section~\ref{sec:desiderata}: robustness to adversarial manipulation, rater
identification, and discriminative power. Real data has no known ground-truth
$\beta_j$. We therefore generate synthetic data with known parameters to
measure absolute recovery accuracy. Once the data-generating process is fixed,
we can ask directly whether \QSMF{} recovers note quality more faithfully than
the baseline, whether it identifies bad raters through \(\hat\rho_i\), and how
these quantities change as the corrupted-rater fraction increases. This
experiment provides the cleanest test of Proposition~\ref{prop:efficiency}:
when ground truth is known and bad raters genuinely have $\rho_i^* = 0$,
heterogeneous weighting should yield strictly better $\beta$ recovery, with the
advantage growing as the fraction of bad raters increases.

We generate synthetic ratings on the real CN sparsity pattern (44,985,977 observations, 412,381 raters, 365,431 notes), calibrating location and scale from the real baseline MF fit. Specifically, we set $\mu = 0.585$, match the standard deviations of $\alpha_i$, $\beta_j$, $\gamma_i$, and $\delta_j$ to the real estimates, and draw each of these latent variables from mean-zero uniform distributions. Ratings are binary:
\(
r_{ij} = \mathbf{1}[\mu + \alpha_i + \rho_i \beta_j + \gamma_i \delta_j + \epsilon_{ij} > 0.5].
\)
Per-rater noise is drawn independently from $\sigma_i \sim \mathrm{Uniform}(0.1, 0.4)$.

\paragraph{Rater types.} A fraction $f_{\mathrm{bad}}$ of raters are designated ``bad.'' Good raters have $\rho_i = 1$. Bad raters have $\rho_i = 0$ and are split into four behavioral types: one third partisan raters who respond only to ideology, one third random raters, one sixth always-helpful raters, and one sixth always-not-helpful raters.

\paragraph{Metric.} We report the $z$-scored MSE, $\mathrm{MSE}_z(\hat{\beta}, \beta^*)$, where both the estimate and the ground truth are standardized to zero mean and unit variance before comparison.

For each $f_{\mathrm{bad}} \in \{0, 0.1, 0.2, 0.3, 0.4, 0.5\}$, we run 10 seeds.
Table entries are mean $\pm$ standard deviation across the 10 seeds.
\begin{center}
\begin{tabular}{lcccc}
\toprule
$f_{\mathrm{bad}}$ & Baseline $\mathrm{MSE}_z$ & \QSMF{} $\mathrm{MSE}_z$ & $\Delta \mathrm{MSE}_z$ & AUC$(\hat{\rho})$ \\
\midrule
0.0 & $0.107 \pm 0.001$ & $0.101 \pm 0.001$ & $0.006 \pm 0.000$ & --- \\
0.1 & $0.128 \pm 0.001$ & $0.111 \pm 0.001$ & $0.016 \pm 0.000$ & $0.949 \pm 0.001$ \\
0.2 & $0.156 \pm 0.001$ & $0.125 \pm 0.001$ & $0.031 \pm 0.000$ & $0.954 \pm 0.001$ \\
0.3 & $0.194 \pm 0.001$ & $0.143 \pm 0.001$ & $0.051 \pm 0.000$ & $0.959 \pm 0.000$ \\
0.4 & $0.247 \pm 0.002$ & $0.168 \pm 0.001$ & $0.080 \pm 0.001$ & $0.963 \pm 0.000$ \\
0.5 & $0.324 \pm 0.002$ & $0.204 \pm 0.001$ & $0.121 \pm 0.002$ & $0.967 \pm 0.000$ \\
\bottomrule
\end{tabular}
\end{center}

\QSMF{} improves $\beta$ recovery at every
bad-rater fraction and in all 10 seeds at each fraction. The gain is already
tightly positive at $f_{\mathrm{bad}} = 0$, then widens in favor of \QSMF{} as corruption
increases. Rater
identification is also strong and stable across seeds: once bad raters are
present, AUC for separating good from bad raters using $\hat{\rho}_i$ ranges
from $0.948$--$0.950$ at $f_{\mathrm{bad}} = 0.1$ up to $0.966$--$0.967$ at
$f_{\mathrm{bad}} = 0.5$. In an illustrative stratification at
$f_{\mathrm{bad}} = 0.3$, \QSMF{} improves recovery in every note-ideology
quintile, so the gain is not confined to ideologically centrist notes. Taken
together, these synthetic results complement the real-data benchmarks above by
showing that the advantages of \QSMF{} are not merely self-referential: when
ground truth is known, it recovers note quality more faithfully while also
identifying corrupted raters more effectively. We turn next to the broader
limitations and implications of these results.

\section{Discussion and Limitations}\label{sec:limitations}

\subsection{Reputation as Incentive and Attack Surface}

The parameter $\rho$ may be interpreted as a reputation, and a well-designed reputation system may 
generate positive externalities for the platform. By making quality-tracking
behavior legible and consequential, it can gamify information aggregation in a
constructive way: contributors may invest more effort, rate more carefully, and
learn that consistently informative ratings increase their influence. Over
time, this may also encourage self-selection, with contributors concentrating
on notes and topics where they are knowledgeable rather than rating
indiscriminately. These behavioral responses are not directly measured in our
experiments, but if they occur they would improve platform information quality
beyond the direct effect captured by the model.

At the same time, any rater-level reputation parameter creates an incentive to
manipulate it. An adversary who understands the $\hat\rho_i$ update rule could,
in principle, rate honestly on many notes to inflate their reputation and then
exploit it on a small number of targets.
More fundamentally, coordinated groups can shift $\hat\beta_j$ itself, changing the benchmark against which all reputations are measured.
We do not claim that $\rho_i$ eliminates all strategic incentives.
However, the most prevalent real-world attack vectors (partisan bloc
voting, random flooding, and constant-rating spam) are precisely the
behaviours that cause $\hat\rho_i$ to collapse
(\S\ref{sec:individual-attacks}).
Without quality sensitivity, these attackers contribute to the consensus
at full weight.
With it, their influence is attenuated automatically and without requiring
any external labels or manual moderation.

\subsection{Additional Parameters and the Bias--Variance Tradeoff}
\QSMF{} adds one scalar parameter $\rho_i$ per rater, increasing model
complexity.
In standard bias--variance terms, this raises variance.
Three features of the design mitigate this cost.
First, the regularization shrinks estimates toward one, preventing the extra parameters from fitting noise.
Second, the two-timescale optimization
(\S\ref{sec:training}) stabilises the joint estimation by
updating reputations on a slower schedule than note-level parameters,
avoiding the feedback loops that would amplify variance in a na\"ive
alternating scheme.
Third, and most importantly, the empirical evidence consistently shows
that the bias reduction dominates:
out-of-sample prediction improves (\S\ref{sec:oos}),
sample efficiency increases (\S\ref{sec:sample-efficiency}),
and synthetic-data recovery is sharper (\S\ref{sec:synthetic}),
all of which would deteriorate if the extra parameters were merely adding
variance.

The regularization parameter $\lambda_\rho$ is also a direct tuning knob on
this tradeoff: decreasing $\lambda_\rho$ permits a broader $\hat\rho_i$
distribution and therefore more aggressive attenuation of noisy or
strategic raters, whereas increasing $\lambda_\rho$ keeps the model closer
to the baseline equal-weighting scheme. Our adopted $\lambda_\rho=0.02$ is
intentionally conservative in this sense. 
The prior center is itself a
design choice. In this paper, centering the penalty at $\rho_i=1$ means
that new or low-activity raters begin with baseline-equivalent influence,
and are downweighted only when the data indicate weak quality tracking. An
alternative would be to center the penalty at $\rho_i=0$, which has a
different interpretation: raters must earn influence on the quality channel
through demonstrated performance. Relative to the $\rho_i=1$ prior, that
choice would be expected to trade some sample efficiency, especially early
in a note's life, for greater default resistance to manipulation and noisy
raters. Both the prior center and the regularization strength
can therefore be calibrated to the application's tolerance for delay,
noise, and manipulation.

\subsection{Ethics of Differential Weighting}
Assigning unequal influence to raters raises a legitimate normative
concern: who decides which voices count more?
We frame this as a \emph{crowdsourced information aggregation} problem
rather than a voting problem.
In information aggregation, weighting by demonstrated accuracy is standard
practice and well-studied \cite{dawid1979, whitehill2009, raykar2010, karger2011, miller2005, shnayder2016}.
\QSMF{} follows the same principle.
Quality sensitivity does not penalise ideological raters per se; it
penalises raters whose ratings carry no quality signal beyond what
ideology already explains.
This distinction is reflected empirically in the small negative correlation
between $\hat\rho_i$ and $|\hat\gamma_i|$
(\S\ref{sec:learned-params}).

A related concern is whether $\rho_i$ captures genuine quality
sensitivity or merely \emph{conformity} with the majority.
The distinction rests on the two-channel structure of the model.
The ideology channel $\gamma_i \delta_j$ already captures subjective
ideological agreement; what remains in the quality channel is the
component of each rating that tracks helpfulness \emph{across} the
ideological spectrum.
A rater who disagrees with the emerging quality consensus because of
ideology will have that disagreement absorbed by $\gamma_i \delta_j$,
not penalised through $\rho_i$.
A rater who disagrees for non-ideological reasons (for instance, a
domain expert who correctly identifies a flaw that most raters miss) will
have high $\rho_i$ once subsequent ratings vindicate their assessment,
assuming $\hat\beta_j$ itself eventually converges toward the expert's position as the true information spreads in the social network.

\subsection{Future Work}

A natural next step is to formalise the attack-resistance properties of
\QSMF{}, in particular for coordinated strategies. 
Another promising extension is to make reputation dynamic through
time-decayed updates to $\rho_i$, so that older ratings gradually receive
less weight when estimating current influence. Such a design could create a
clearer path to redemption for raters whose past behavior was low-quality but
later improves, while also strengthening incentives for continuous
constructive participation rather than one-time reputation building. It would
also be valuable to test \QSMF{} on data from additional platforms that have
adopted Community Notes-style programmes (YouTube, Meta, TikTok), where the
rater population and content mix may differ substantially. Finally, the
regularisation strength $\lambda_\rho$ and prior centre are currently fixed;
an adaptive scheme that tunes these based on the observed rater pool could
improve the bias--variance tradeoff across different deployment regimes.

\bibliographystyle{plain}
\bibliography{refs}

\appendix

\section{Proofs}\label{sec:proofs}

\subsection{Exact Conditional Note-Block Update}\label{app:conditional-note-update}

Fix a note \(j\), and hold all parameters except \((\beta_j,\delta_j)\) fixed.
Define the partial residual
\(
\hat y_{ij} := r_{ij} - \hat\mu - \hat\alpha_i,
\quad i \in \mathcal I_j.
\)
Then the note-\(j\) subproblem is
\begin{equation}
\label{eq:note-subproblem}
\min_{\beta,\delta}\;
\frac12 \sum_{i\in\mathcal I_j}
\bigl(\hat y_{ij} - \hat\rho_i \beta - \hat\gamma_i \delta \bigr)^2
+ \frac{\lambda_\beta}{2}\beta^2
+ \frac{\lambda_\delta}{2}\delta^2.
\end{equation}
Because \(\lambda_\beta,\lambda_\delta>0\), this objective is strictly convex
in \((\beta,\delta)\) and therefore has a unique minimizer. Let
\(
\hat S_j := \sum_{i\in\mathcal I_j}\hat\rho_i^2,
\qquad
\hat V_j := \sum_{i\in\mathcal I_j}\hat\gamma_i^2,
\qquad
\hat C_j := \sum_{i\in\mathcal I_j}\hat\rho_i\hat\gamma_i .
\)

The baseline model has the same exact conditional note-block update, specialized
to the case \(\rho_i \equiv 1\) for all raters \(i\). In that case,
\(\hat S_j = n_j\) and
\(
\hat C_j = \sum_{i\in\mathcal I_j}\hat\gamma_i.
\)

\begin{proposition}[Conditional note-block update]
\label{prop:conditional-note-update}
For fixed \((\hat\mu,\hat\alpha,\hat\rho,\hat\gamma)\), the unique minimizer of
\eqref{eq:note-subproblem} is
\begin{equation}
\label{eq:conditional-note-update}
\begin{bmatrix}
\hat\beta_j\\[3pt]
\hat\delta_j
\end{bmatrix}
=
\begin{bmatrix}
\hat S_j+\lambda_\beta & \hat C_j\\
\hat C_j & \hat V_j+\lambda_\delta
\end{bmatrix}^{-1}
\begin{bmatrix}
\sum_{i\in\mathcal I_j}\hat\rho_i \hat y_{ij}\\[3pt]
\sum_{i\in\mathcal I_j}\hat\gamma_i \hat y_{ij}
\end{bmatrix}.
\end{equation}
Equivalently,
\begin{equation}
\label{eq:beta-update-general}
\hat\beta_j
=
\frac{(\hat V_j+\lambda_\delta)\sum_{i\in\mathcal I_j}\hat\rho_i \hat y_{ij}
-
\hat C_j\sum_{i\in\mathcal I_j}\hat\gamma_i \hat y_{ij}}
{(\hat S_j+\lambda_\beta)(\hat V_j+\lambda_\delta)-\hat C_j^2}.
\end{equation}
\end{proposition}

\begin{proof}
The objective in \eqref{eq:note-subproblem} is a quadratic in \((\beta,\delta)\).
Its first-order conditions are
\[
\sum_{i\in\mathcal I_j}\hat\rho_i(\hat y_{ij}-\hat\rho_i\beta-\hat\gamma_i\delta)-\lambda_\beta \beta = 0,
\]
\[
\sum_{i\in\mathcal I_j}\hat\gamma_i(\hat y_{ij}-\hat\rho_i\beta-\hat\gamma_i\delta)-\lambda_\delta \delta = 0.
\]
Rearranging gives the linear system
\[
\begin{bmatrix}
\hat S_j+\lambda_\beta & \hat C_j\\
\hat C_j & \hat V_j+\lambda_\delta
\end{bmatrix}
\begin{bmatrix}
\beta\\
\delta
\end{bmatrix}
=
\begin{bmatrix}
\sum_{i\in\mathcal I_j}\hat\rho_i \hat y_{ij}\\[3pt]
\sum_{i\in\mathcal I_j}\hat\gamma_i \hat y_{ij}
\end{bmatrix}.
\]
By Cauchy--Schwarz, \(\hat C_j^2 \le \hat S_j \hat V_j\), so
\[
(\hat S_j+\lambda_\beta)(\hat V_j+\lambda_\delta)-\hat C_j^2
\ge
\hat S_j\lambda_\delta+\hat V_j\lambda_\beta+\lambda_\beta\lambda_\delta
>0.
\]
Hence the coefficient matrix is positive definite and therefore invertible,
yielding \eqref{eq:conditional-note-update}. The expression
\eqref{eq:beta-update-general} is the first coordinate of this inverse.
\end{proof}

\subsection{Proof of Proposition~\ref{prop:efficiency}}\label{app:proof-efficiency}

Using \eqref{eq:oracle-reduced-model},
\[
\mathbb E[d^*_{ij}]=\rho_i^*\beta_j^*,
\qquad
\mathrm{Var}(d^*_{ij})=\sigma^2.
\]
Therefore
\[
\mathbb E[\hat\beta_j^{QS}]
=
\frac{\sum_i \rho_i^*\, \mathbb E[d^*_{ij}]}
{\sum_i (\rho_i^*)^2}
=
\frac{\sum_i (\rho_i^*)^2 \beta_j^*}
{\sum_i (\rho_i^*)^2}
=
\beta_j^*,
\]
so \(\hat\beta_j^{QS}\) is unbiased. By the cross-rater independence assumption
in Section~\ref{sec:efficiency},
\[
\mathrm{Var}(\hat\beta_j^{QS})
=
\frac{\sum_i (\rho_i^*)^2 \sigma^2}
{\left(\sum_i (\rho_i^*)^2\right)^2}
=
\frac{\sigma^2}{\sum_i (\rho_i^*)^2}.
\]

Now consider an equal-weight estimator
\[
\tilde\beta_j=c_j\sum_i d^*_{ij}.
\]
Its expectation is
\[
\mathbb E[\tilde\beta_j]
=
c_j \sum_i \rho_i^*\beta_j^*
=
c_j n_j \bar\rho_j^* \beta_j^*.
\]
Thus unbiasedness requires
\[
c_j=\frac{1}{n_j\bar\rho_j^*},
\]
which is unique when \(\bar\rho_j^*\neq 0\). Its variance is then
\[
\mathrm{Var}(\tilde\beta_j^{\mathrm{uni}})
=
\frac{1}{n_j^2(\bar\rho_j^*)^2}
\sum_i \mathrm{Var}(d^*_{ij})
=
\frac{1}{n_j^2(\bar\rho_j^*)^2}\cdot n_j \sigma^2
=
\frac{\sigma^2}{n_j(\bar\rho_j^*)^2}.
\]
\qed

\subsection{Derivation of the \texorpdfstring{$\rho_i$}{rho_i} Update}\label{app:rho-update-derivation}

The objective in \eqref{eq:rho-subproblem} is a strictly convex quadratic in
\(\rho\), so it has a unique minimizer on \([0,\infty)\). Expanding the objective,
\[
\frac12\sum_{j\in\mathcal N_i}(\hat d_{ij}-\rho\hat\beta_j)^2
+\frac{\lambda_\rho}{2}(\rho-1)^2
=
\frac12 D_i + \frac{\lambda_\rho}{2} -(A_i+\lambda_\rho)\rho+\frac12(B_i+\lambda_\rho)\rho^2.
\]
The unconstrained minimizer is therefore \((A_i+\lambda_\rho)/(B_i+\lambda_\rho)\). Projecting
onto the feasible set \([0,\infty)\) gives
\[
\hat\rho_i=\max\!\left\{0,\frac{A_i+\lambda_\rho}{B_i+\lambda_\rho}\right\}
=\frac{(A_i+\lambda_\rho)_+}{B_i+\lambda_\rho}.
\]
\qed

\subsection{Score Component Derivation}\label{app:score-test}

This appendix provides the details behind the score-based diagnostic in
\S\ref{sec:score-test}.

Consider the unrestricted model
\(
r_{ij}
=
\mu + \alpha_i + \rho_i\beta_j + \gamma_i\delta_j + \epsilon_{ij},
\quad
\epsilon_{ij}\sim\mathcal N(0,\sigma^2).
\)
The baseline Community Notes model corresponds to the restriction
\(\rho_i=1\) for all~\(i\). Under Gaussian errors, the log-likelihood
contribution from the ratings of rater~\(i\) is (up to an additive constant)
\[
\ell_i(\rho_i)
=
-\frac{1}{2\sigma^2}
\sum_{j\in\mathcal N_i}
(r_{ij}-\mu-\alpha_i-\rho_i\beta_j-\gamma_i\delta_j)^2.
\]
The partial derivative with respect to \(\rho_i\) is
\[
\frac{\partial\ell_i}{\partial\rho_i}
=
\frac{1}{\sigma^2}
\sum_{j\in\mathcal N_i}
\beta_j\,(r_{ij}-\mu-\alpha_i-\rho_i\beta_j-\gamma_i\delta_j).
\]
Evaluating at the restricted estimate \(\rho_i=1\) and the fitted baseline
parameters, the residual becomes
\(
e_{ij}^{\text{base}}
= r_{ij}
  -\hat\mu^{\text{base}}
  -\hat\alpha_i^{\text{base}}
  -\hat\beta_j^{\text{base}}
  -\hat\gamma_i^{\text{base}}\hat\delta_j^{\text{base}},
\)
so the score component reduces to
\(
U_i
=
\sum_{j\in\mathcal N_i}
\hat\beta_j^{\text{base}}\, e_{ij}^{\text{base}},
\)
which is the Rao score component for \(\rho_i\) (up to the factor
\(1/\sigma^2\)). Under the null, \(\E[U_i]=0\).

\paragraph{Slope interpretation.}
As shown in Section~\ref{sec:score-test}, the identity
\(d_{ij}^{\text{base}} = e_{ij}^{\text{base}} + \hat\beta_j^{\text{base}}\)
implies
\(
s_i - 1
=
U_i\big/\!\sum_{j\in\mathcal N_i}(\hat\beta_j^{\text{base}})^2,
\)
so the slope statistic \(s_i\) is an affine transformation of the score
component for~\(\rho_i\).


\section{Replication with Two Ideology Dimensions}\label{app:k2}

\begin{figure}[h]
\centering
\includegraphics[width=0.8\linewidth]{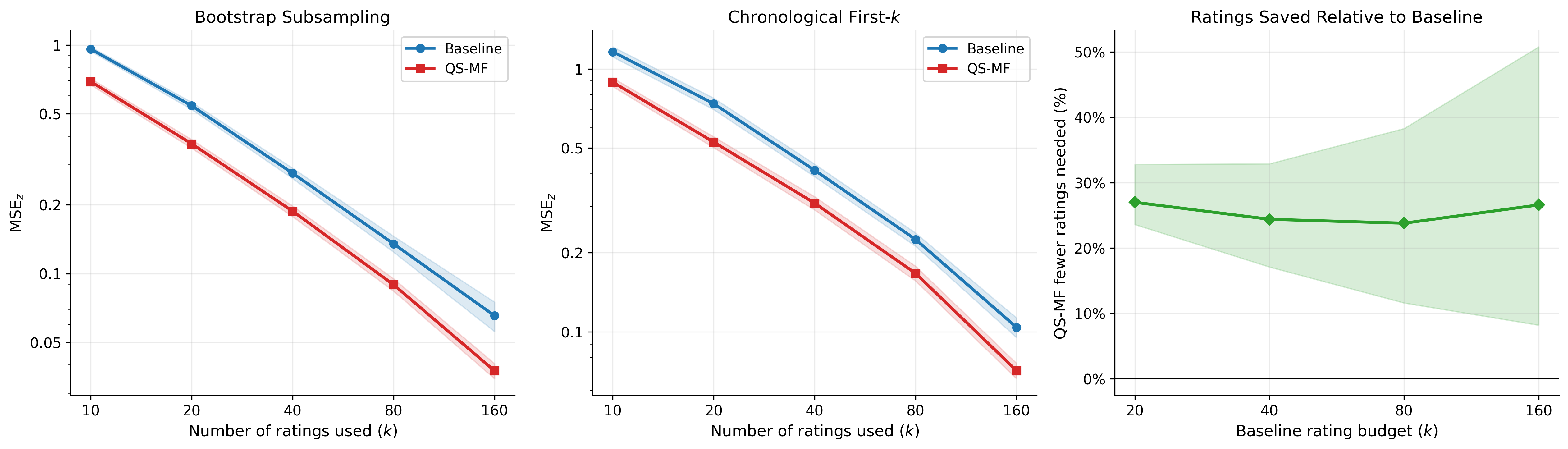}
\caption{Rolling-window sample efficiency with two ideology dimensions. The design matches Figure~\ref{fig:rolling-window}; only the ideology term is expanded to two latent dimensions. The band shows the range across the 14 evaluation windows after averaging ratings-saved within each window.}
\label{fig:rolling-window-k2}
\end{figure}

The main experiments use a single ideology factor ($k=1$), matching the production CN system. A natural concern is that the gains from \QSMF{} are an artifact of the one-dimensional ideology assumption: perhaps a richer ideology model would absorb the variation that $\rho_i$ currently captures. We address this by repeating the out-of-sample prediction and rolling-window sample efficiency experiments with $k=2$ ideology dimensions.

\subsection{Out-of-Sample Prediction}\label{app:oos-k2}

The quality-sensitivity channel remains active within the $k=2$ specification: \QSMF{} reduces held-out MSE from $0.071$ to $0.069$, a $2.23\% \pm 0.02\%$ reduction, and wins in all 10 splits. Thus the main conclusion from Section~\ref{sec:oos} survives the higher-dimensional ideology model.

\subsection{Rolling-Window Sample Efficiency}\label{app:rolling-k2}

We next repeat the rolling-window design from Section~\ref{sec:rolling-window} with $k=2$ ideology dimensions. The cohort rule, cutoff schedule, future-consensus target, frozen-rater construction, and evaluation metrics are unchanged, so we only report the aggregate results here and refer to the main text for the full design and interpretation.

Table entries are mean $\pm$ standard deviation across the 14 window-level aggregates.
\begin{center}
\begin{tabular}{rccccc}
\toprule
& \multicolumn{2}{c}{Bootstrap} & \multicolumn{2}{c}{Temporal} & Ratings saved \\
$k$ & Baseline & \QSMF{} & Baseline & \QSMF{} & (bootstrap) \\
\midrule
10  & $0.961 \pm 0.041$ & $0.692 \pm 0.042$ & $1.161 \pm 0.098$ & $0.890 \pm 0.065$ & --- \\
20  & $0.543 \pm 0.038$ & $0.370 \pm 0.031$ & $0.737 \pm 0.071$ & $0.527 \pm 0.048$ & $27.0\%$ \\
40  & $0.275 \pm 0.028$ & $0.188 \pm 0.019$ & $0.412 \pm 0.043$ & $0.309 \pm 0.034$ & $24.4\%$ \\
80  & $0.135 \pm 0.021$ & $0.090 \pm 0.011$ & $0.225 \pm 0.025$ & $0.167 \pm 0.021$ & $23.8\%$ \\
160 & $0.066 \pm 0.019$ & $0.038 \pm 0.006$ & $0.104 \pm 0.017$ & $0.071 \pm 0.009$ & $26.6\%$ \\
\bottomrule
\end{tabular}
\end{center}

The qualitative pattern is unchanged. In the bootstrap design, relative improvement ranges from $28.0\%$ at $k=10$ to $42.5\%$ at $k=160$; in the chronological design, it ranges from $23.4\%$ to $31.6\%$. The baseline-anchored ratings-saved measure also remains positive, averaging $27.0\%$, $24.4\%$, $23.8\%$, and $26.6\%$ at $k = 20, 40, 80$, and $160$, respectively. Thus the main sample-efficiency conclusion is not specific to the one-dimensional ideology model.

\end{document}